\documentclass[a4paper,UKenglish]{lipics}
\usepackage{amsmath}
\usepackage{amssymb}
\usepackage{amsthm}
\usepackage{graphicx}
\usepackage{color}

\usepackage{algpseudocode,algorithm,algorithmicx}
\usepackage[numbers,sort&compress]{natbib}

\usepackage{enumerate}

\usepackage{enumitem}
\usepackage[labelfont=bf,width=.85\textwidth,font=small]{caption}

\usepackage{thmtools}
\usepackage{thm-restate}

\newcommand{\canout}{canonical outgoing }

\newcommand{\SP}{\ensuremath{\mathcal{S}}}
\newcommand{\Or}{\ensuremath{\mathcal{O}}}

\newcommand{\out}[1]{\mathrm{out\mbox{-}index}(#1)}
\newcommand{\ini}[1]{\mathrm{in\mbox{-}index}(#1)}
\newcommand{\ind}[1]{\mathrm{index}(#1)}
\newcommand{\lind}[1]{\mbox{$\mathrm{layer}$-$\mathrm{index}(#1)$}}

\newcommand{\C}{\ensuremath{\mathcal C}}
\newcommand{\LL}{\ensuremath{\mathcal L}}
\newcommand{\PP}{\ensuremath{\mathcal P}}

\graphicspath{{figures/},{}}

\begin{document}


%
%
%

\begin{titlepage}


%
%
%

{\centering \LARGE \noindent Peeling and Nibbling the Cactus: \\[0.2cm]  Subexponential-Time Algorithms for \\[0.2cm] Counting Triangulations and Related Problems.\\ }


\vfill

\begin{figure}[htbp]
  \centering
  \includegraphics{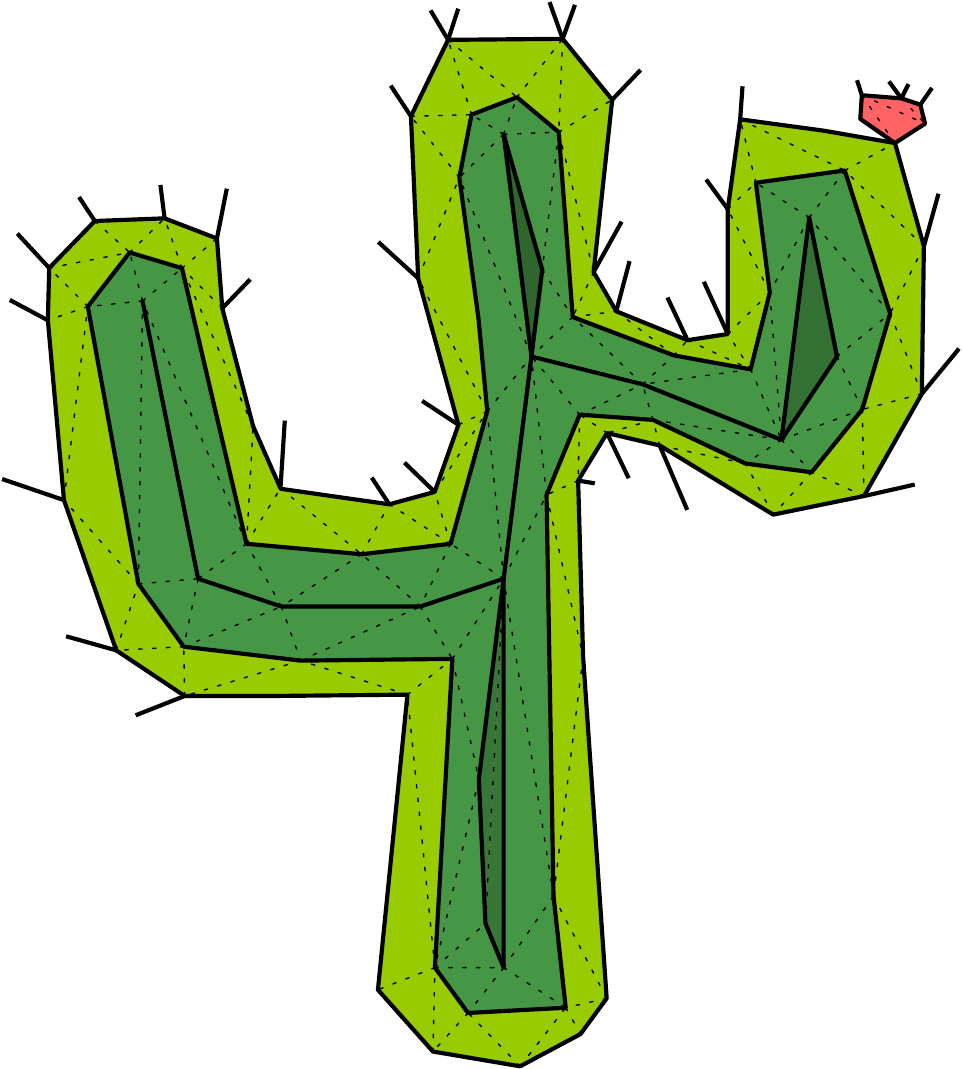}
\end{figure}

\vfill


\begin{center}
{\Large D\'{a}niel Marx  \quad Tillmann Miltzow }
\end{center}

\begin{center}
Institute for Computer Science and Control,\\
Hungarian Academy of Sciences (MTA SZTAKI)\\
\texttt{ dmarx@cs.bme.hu}, \texttt{t.miltzow@gmail.com}

\end{center}

%

\end{titlepage}


\begin{abstract}
  Given a set of $n$ points $S$ in the plane, a triangulation $T$ of
  $S$ is a maximal set of non-crossing segments with endpoints in $S$.
  We present an algorithm that computes the number of triangulations
  on a given set of $n$ points in time $n^{(11+ o(1))\sqrt{n} }$,
  significantly improving the previous best running time of
  $O(2^n n^2)$ by Alvarez and Seidel [SoCG 2013].  Our main tool is
  identifying separators of size $O(\sqrt{n})$ of a triangulation in a
  canonical way. The definition of the separators are based on the
  decomposition of the triangulation into nested layers (``cactus
  graphs'').  Based on the above algorithm, we develop a simple and
  formal framework to count other non-crossing straight-line graphs in
  $n^{O(\sqrt{n})}$ time.  We demonstrate the usefulness of 
  the framework by applying
  it to counting non-crossing Hamilton cycles, spanning trees, perfect
  matchings, $3$-colorable triangulations, connected graphs, cycle
  decompositions, quadrangulations, $3$-regular graphs, and more.
\end{abstract}

\section{Introduction}

 Given a set of $n$ points in the plane, a triangulation $T$ of $S$ is defined to be a maximal set of non-crossing line segments with both endpoints in $S$. This set of segments together with the set $S$ defines a plane graph. It is easy to see that every bounded face of a triangulation $T$ is indeed a triangle. 
 We assume that $S$ is in general position: no three points of $S$ are on a line.
Triangulations are one of the most studied concepts in discrete and computational geometry, studied both from combinatorial and algorithmic perspectives 
\cite{bern1992mesh, 
DBLP:journals/dcg/Chazelle91, 
de2000computational, 
edelsbrunner2010computational,   
DBLP:journals/jct/Fisk78a,
DBLP:journals/algorithmica/GuibasHLST87,
hershberger1995pedestrian,
DBLP:journals/combinatorics/SharirS11}.
It is well known that the number of possible triangulations of $n$ points in convex position is exactly the $(n-2)$-th Catalan number, but counting the number of triangulations of arbitrary point sets seems to be a much harder problem. There is a long line of research devoted to finding better and better exponential-time algorithms for counting triangulations \cite{
aichholzer1999path,
DBLP:journals/dcg/AlvarezBCR15,
DBLP:journals/corr/AlvarezBR13,
DBLP:journals/comgeo/AlvarezBRS15,
DBLP:conf/compgeom/AlvarezS13,
anagnostou1993polynomial,
avis1996reverse,
dumitrescu2001enumerating,
gilbert1979new,
DBLP:conf/icalp/KarpinskiLS15,
klincsek1980minimal,
ray2004simple}.
 The sequence of improvements culminated in the $O(2^n n^2)$ time algorithm of Alvarez and Seidel \cite{DBLP:conf/compgeom/AlvarezS13}, winning the best paper award at SoCG 2013. Our main result significantly improves the running time of counting triangulations by making it subexponential:
\begin{theorem}[General Plane Algorithm]\label{thm:FullPlaneAlgo}
   There exists an algorithm that given a set $S$ of $n$ points in the plane
  computes the number of all triangulations of $S$ in $n^{(11+o(1))\sqrt{n}}$
  time.
\end{theorem}

It is very often the case that restricting an algorithmic problem to
planar graphs allows us to solve it with much better worst-case
running time than what is possible for the unrestricted problem. One
can observe a certain ``square root phenomenon'': in many cases, the
best known running time for a planar problem contains a square root in
the exponent. For example, the 3-Coloring problem on an $n$-vertex
graph can be solved in subexponential time $2^{O(\sqrt{n})}$ on planar
graphs (e.g., by observing that a planar graph on $n$ vertices has
treewidth $O(\sqrt{n})$\,), but only $2^{O(n)}$ time algorithms are
known for general graphs. Moreover, it is known that if we assume the
Exponential-Time Hypothesis (ETH), which states that there is no
$2^{o(n)}$ time algorithm for $n$-variable 3SAT, then there is no
$2^{o(\sqrt{n})}$ time algorithm for 3-Coloring on planar graphs and
no $2^{o(n)}$ time algorithm on general graphs
\cite{DBLP:journals/eatcs/LokshtanovMS11}.  The situation is similar
for the planar restrictions of many other NP-hard problems, thus it
seems that the appearance of the square root of the running time is an
essential feature of planar problems.  A similar phenomenon occurs in
the framework of parameterized problems, where running times of the
form $2^{O(\sqrt{k})}\cdot n^{O(1)}$ or $n^{O(\sqrt{k})}$ appear for
many planar problems and are known to be essentially best possible
(assuming ETH)
\cite{DBLP:journals/siamdm/DemaineFHT04,DBLP:journals/jacm/DemaineFHT05,DBLP:journals/talg/DemaineFHT05,DBLP:journals/siamcomp/FominT06,DBLP:journals/csr/DornFT08,DornPBF10,DBLP:conf/gd/DemaineH04,DBLP:journals/combinatorica/DemaineH08,DBLP:journals/cj/DemaineH08,DBLP:conf/esa/Thilikos11,
  DBLP:journals/ipl/FominLRS11,DBLP:conf/stacs/DornFLRS10,DBLP:conf/icalp/KleinM12,DBLP:conf/soda/KleinM14,DBLP:conf/soda/ChitnisHM14,DBLP:conf/stacs/PilipczukPSL13,DBLP:conf/focs/PilipczukPSL14,DBLP:conf/esa/MarxP15}.

A triangulation of $n$ points can be considered as a planar graph on $n$ vertices, hence it is a natural question whether the square root phenomenon holds for the problem of counting triangulations. Indeed, for the related problem of finding a minimum weight triangulation, subexponential algorithms with running time $n^{O(\sqrt{n})}$ are known \cite{DBLP:conf/wg/KnauerS06,DBLP:conf/cccg/Lingas98}. 
These algorithms are based on the use of small balanced separators. Given a plane triangulation on $n$ points in the plane, it is well known that there exists a balanced $O(\sqrt{n})$-sized separator that divides the triangulation into at least two independent graphs~\cite{lipton1979separator}. 
The basic idea is to guess a correct $O(\sqrt{n})$-sized separator of a minimum weight triangulation and recurse one all occurring subproblems. 
As there are only $n^{O(\sqrt{n})}$ potential graphs on $O(\sqrt{n})$ vertices, one can show that the whole algorithm takes $n^{O(\sqrt{n})}$ time~\cite{DBLP:conf/wg/KnauerS06,DBLP:conf/cccg/Lingas98}.

Unfortunately, this approach has serious problems when we try to apply it to counting triangulations. The fundamental issue with this approach is that a triangulation of course may have more than one $O(\sqrt{n})$-sized balanced separators and hence we may overcount the number of triangulations, as a triangulation would be taken into account 
in more than one of the guesses. To get around this problem, an obvious simple idea would be to say that we always try to guess a ``canonical'' separator, for example, the lexicographically first separator.  However, it is a complete mystery how to guarantee in subsequent recursion steps that the separator we have chosen is indeed the lexicographic first for all the triangulations we want to count.
Perhaps the most important technical idea of the paper is finding a suitable way of making the separators canonical. 

This full version is committed to give all details in a self-contained way.
The reader who is only interested in the main concepts of the algorithm is refered to the conferenc version~\cite{PeelingCactusAppear}. 

\subsection{Preliminaries}\label{sec:Prelim}

We interpret a collection of points $S$ and non-crossing segments $E$ in $\mathbb{R}^2$ as a plane graph, if every segment $e\in E$ shares exactly its endpoints with $S$. We usually identify points with vertices and edges with segments. 
For convenience, we will use the term graph almost always,
to indicate plane straight line graphs and silently assume we mean a plane straight line graph. 
The number $n$ always denotes the size of the underlying point set $S$.
We assume $S$ to be in general position, that is, no three points lie on a line and no four points lie on a common circle.

\begin{figure}[t]
  \centering
  \includegraphics[width = 0.8\textwidth]{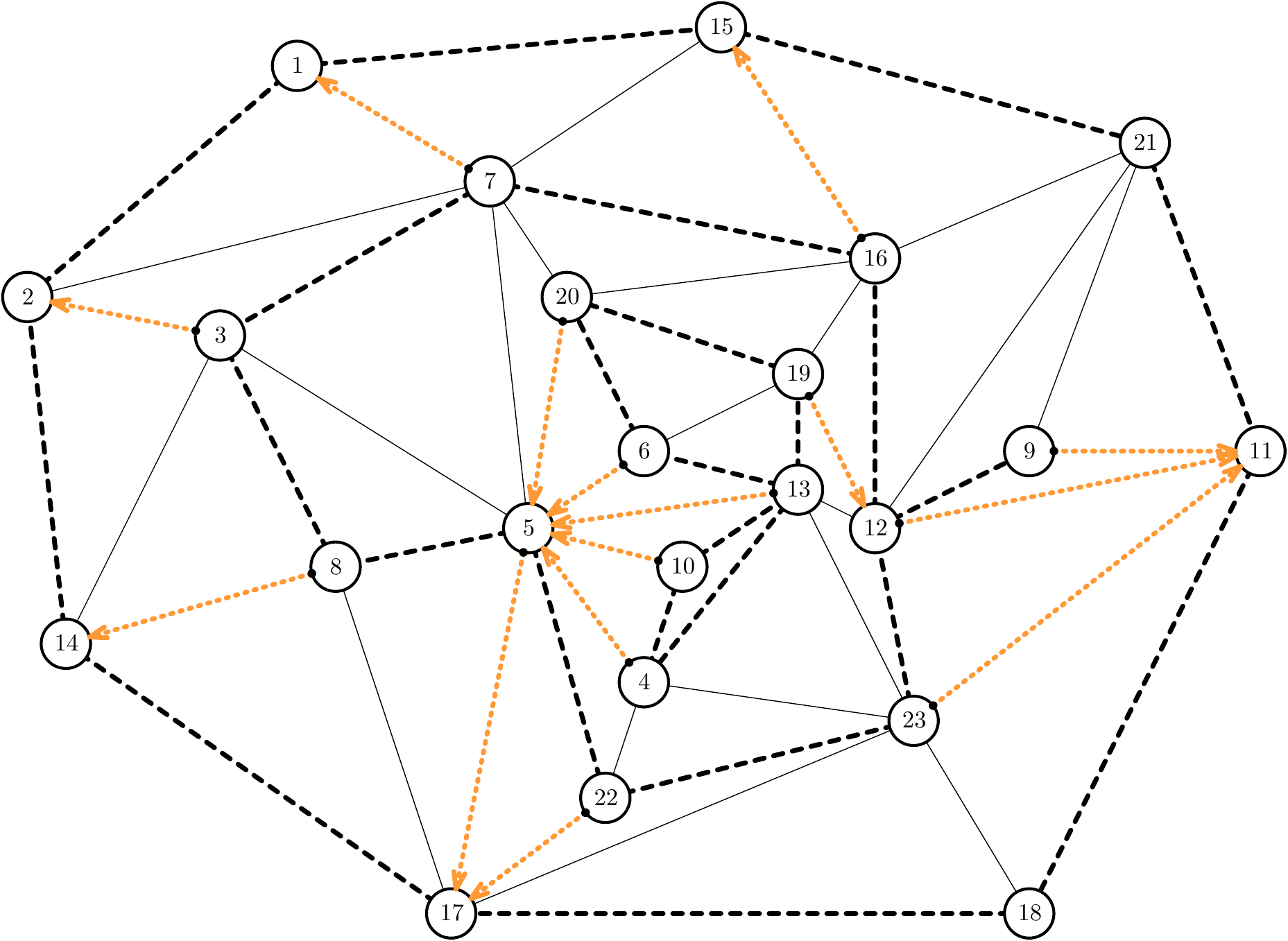}
  \caption{A set of points, each with a distinct \emph{order label}, and a triangulation. The edges of layers are dashed. For each vertex, we drew a directed edge to the unique neighbor with smallest order label, among the vertices with smaller distance to the boundary.}
  \label{fig:orderedlabels}
\end{figure}

A plane graph induces naturally a partition of the plane into open faces, open segments and a finite collection of points.
The unique  unbounded face is called the \emph{outer face}.
We call a plane graph $G$ a \emph{cactus graph} (or just \emph{cactus}) if all its vertices and edges are incident to the \emph{outer face}. 
We call an edge \emph{inner edge} of $G$ if it is not on the outer face of $G$.
A cactus graph is outerplane, but some outerplane graphs are not cacti, as they might have inner edges. 
For convenience, we do \emph{not} require cactus graphs to be connected.
A \emph{triangulation} of a set of points is a maximal plane graph on those points.

We define a decomposition of a triangulation into nested \emph{(cactus) layers} of cacti. The first (cactus) layer is defined by the set of vertices and edges incident to the outer face. Inductively, the $i$-th layer is
defined by the vertices and edges incident to the outer faces after the 
first $i-1$ layers are removed and has \emph{index} $i$. 
We say layer $i$ is further \emph{outside} then layer $j$ if $i<j$ and in this case layer $j$ is more \emph{inside} than layer $i$.
The \emph{outerplanar index} of a graph is defined by the number of non-empty (cactus) layers. Further, we can give each vertex $v$ uniquely the \emph{index}, denoted by $\lind{v}$, of the layer it is contained in.

We define $\partial CH(S)$ as the boundary of the convex hull of the point set $S$. The \emph{onion layers} of a set of points $S$ are defined inductively in a similar fashion. The first layer is $\partial CH(S)$. The $i$-th layer is the boundary of the convex hull after the first $i-1$ layers are removed.

The definition of cactus layers and onion layers should not be confused: the onion layers are completely defined by the point set only, whereas cactus layers are defined by the point set and the triangulation. In particular, it is easy to construct a point set with an arbitrary large number of onion layers, but having a triangulation with only two cactus layers.

We attach to each point $p\in S$ a different number $l(p)\in \{1,\ldots, n\}$ and 
refer to this number as the \emph{order label}.

\begin{lemma}\label{lem:LayerDistance}
  Given a valid triangulation $T$ of some set $S$ of $n$ points in the plane.
  Then each vertex $v$ with  $\lind{v} > 1$
  has a neighbor $w$ with $\lind{w} = \lind{v} -1 $,  vertex $v$ has no neighbor $w$ with $\lind{w} < \lind{v} -1 $.
  Among the neighbors $w$ as described above, there is exactly one with smallest order label.
\end{lemma}
\begin{proof}
  Consider a vertex $v$ with $i= \lind{v} > 1$. We denote by $w_1,\ldots,w_k$ the neighbors of $v$ in cyclic order. We will show one of the neighbors has layer-index $\lind{v}-1$.
  We know after $\lind{v}-1$ iterations of removal of vertices on the outer face
  of $G$, $v$ is on the outer face of $G_{i}$. This implies one of the faces adjacent to $v$ becomes part of the outer face of $G_i$. This implies one of 
  the edges $(w_j,w_{j+1})$ is removed (indices taken modulo k.).
  This implies $w_j$ or $w_{j+1}$ was removed from $G_{i-1}$. 
  The index of a removed vertex is by definition $i-1$. This shows the first claim.
  
  To see the second claim, 
  assume there exists a vertex $w$ with the above mentioned index.
  After the removal of $w$ the vertex $v$ 
  is on the outside face and thus $\lind{v} \leq \lind{w}  + 1 < \lind{v} -1  + 1 = \lind{v} $ --- a contradiction.
\end{proof}

In light of Lemma~\ref{lem:LayerDistance}, we can define layers alternatively using the distance to the vertices on the outer face.
For this purpose, define an artificial vertex $v_{\infty}$ adjacent
to all vertices on the outer face.
Let $d(v)$ denote the \emph{distance of $v$ to $v_{\infty}$},
that is the length of the shortest path to $v_\infty$. By Lemma~\ref{lem:LayerDistance}, $d(v) = \lind{v}$,
and layer $i$ has as vertices $V(L_i) = \{v : d(v) = i\}$.
Note that the graph induced on $V(L_i)$ is \emph{not} layer $i$.
Here we use the convention that the length of a path equals 
the number of its edges.


\begin{definition}
  An \emph{annotated} triangle is a $9$-tuple consisting of $3$ points of $S$, which form an empty triangle and $6$ strings, one for each vertex and edge of the triangle.
  An \emph{annotation system} is a list $L$ of annotated triangles. We say a triangle is \emph{feasible}, if it belongs to the list. The size of $L$ is the number of triangles it contains and denoted by $|L|$.

    Given an annotation system $L$, we call the $T$ a \emph{valid annotated triangulation} of $S$ if for every annotated triangle $\Delta$ of $T$ is feasible($\Delta \in L$).
    Further $\mathcal{T}^A({L})$ denotes the set of all valid annotated triangulations belonging to $L$. 
  \end{definition}

  We denote with $[a,b]$ the \emph{integer interval} $\{a,a+1,\ldots,b\}$.

\subsection{Results}

\begin{figure}[t]
  \centering
  \includegraphics{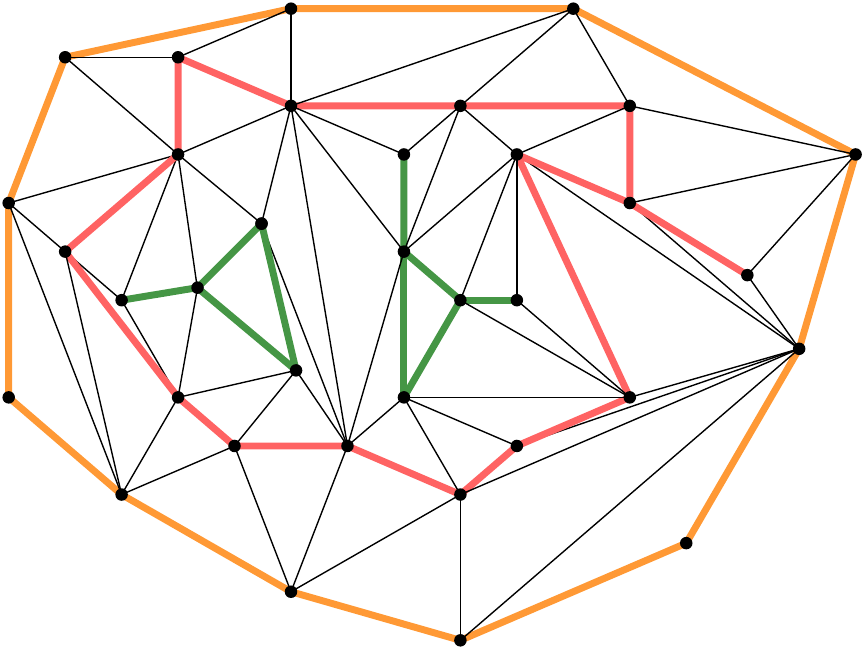}
  \caption{A triangulation with three cactus layers. Layer $1$ is colored orange, Layer two is colored red and Layer $3$ is colored green. Note that the layers do not need to be connected.}
\end{figure}

Given a triangulation $T$, we define small \emph{canonical} separators by distinguishing two cases. If $T$ has more than $\sqrt{n}$ cactus layers, then one of the first $\sqrt{n}$ layers has size at most $\sqrt{n}$ and we can define the one with smallest index to be the canonical separator. 
Using such a separator, we \emph{peel off} some cacti to reduce the problem size.
In the case when we have only a few cactus layers, we can define short canonical separator paths from the interior to the outer face of the triangulation.
We formalize both ideas into a dynamic programming algorithm. The main difficulty is to define the subproblems appropriately. We use the so-called ring subproblems for the layer separators and nibbled rind subproblems for the path separators.

As a byproduct of this algorithmic scheme, we can efficiently count triangulation with a small number of layers. This is similar to previous work on finding a minimum weight triangulation~\cite{anagnostou1993polynomial} and counting triangulations~\cite{DBLP:journals/dcg/AlvarezBCR15} for point sets with a small number of onion layers. 
\begin{restatable}[Thin Plane Algorithm]{theorem}{thmThinPlane}\label{thm:ThinPlane}
  There exists an algorithm  that
  given a set $S$ of $n$ points in the plane
  computes the number of all triangulations of $S$ 
  with outerplanar index $k$ in 
  $n^{O(k)}$ time.
\end{restatable}

One may want to count triangulations subject to certain constraints
(e.g., degree bounds, or bounds on the angles of the triangles, etc.)
or generalize the problem to counting colored triangulations with
colors on the vertices or edges. We introduce an annotated version of
the problem to express such generalizations in a clean and formal way.
An \emph{annotated} triangle is a $9$-tuple consisting of $3$ points
of $S$, which form an empty triangle and $6$ strings, one for each
vertex and edge of the triangle. An \emph{annotation system} is a list
$L$ of annotated triangles. Given an annotation system $L$, we call an
annotated triangulation $T$ \emph{valid} if every annotated triangle
$\Delta$ of $T$ belongs to $L$.  With little extra effort, we can generalize
our algorithms to count also valid annotated triangulations.
We denote by $|L|$ the number of annotated triangles and assume that each string can be described with $n^{O(1)}$ bits.
  \begin{theorem}[Counting Annotated Triangulations]\label{thm:AnnotToAlgo}
      Given an annotation system $L$ and a set $S$ of $n$ points in the plane, we can count all valid annotated triangulation in $n^{(11+o(1))\sqrt{n}}\cdot |L|^{(12+o(1))\sqrt{n}}$
      time.
  \end{theorem}
    
As examples of this generalization, we can count triangulations that are 3-colorable or where each point has a specified degree in the triangulation: all we need is to carefully design a suitable annotation system.
\begin{figure}[t]
  \centering
  \includegraphics{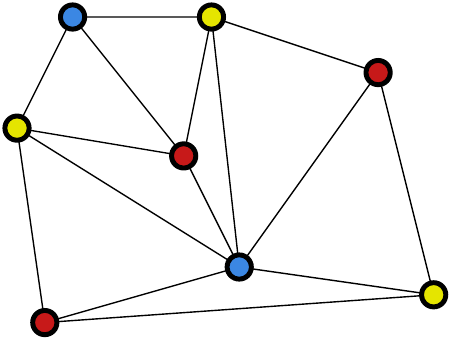}
  \caption{A three colorable triangulation. 
  }
  \label{fig:ThreeColorableTriangulation}
\end{figure}

   \begin{theorem}\label{thm:Count3Color}
    Given a set $S$ of $n$ points in the plane, we can counts all $3$-colorable triangulations of $S$ in time $n^{O(\sqrt{n})}$.
    \end{theorem}

  \begin{theorem}\label{thm:CountDegreeConstrainedTriangulations}
    Given a set $S$ of $n$ points in the plane with prescribed degrees on each vertex, we can count all triangulations $T$ satisfying the degree constraints in  $n^{O(\sqrt{n})}$ time.
  \end{theorem}

More generally, instead of triangulations, we could be interested in counting other geometric graph classes, such as non-crossing perfect matchings, non-crossing Hamilton cycles, etc. Surprisingly, many such problems can be expressed in a completely formal way in our framework of counting annotated triangulations. The idea here is to think of a geometric graph on a point set as a 2-edge-colored triangulation, with one color forming the graph itself and the other color representing non-edges. To make this idea work, we have to ensure that for each member of our graph class, we count only one 2-edge-colored triangulation. This is nontrivial, as a given geometric graph can be extended into a a 2-edge-colored triangulation in many different ways. Similarly to previous work \cite{DBLP:journals/dcg/AlvarezBCR15, DBLP:journals/comgeo/AlvarezBRS15}, we use the notion of constrained Delaunay triangulation (see \cite{hjelle2006triangulations}) to enforce that each graph has a unique extension into a valid 2-edge-colored triangulation.   By formalizing this idea and carefully designing annotation systems, it is possible to get $n^{O(\sqrt{n})}$ time algorithms for a large number of graph classes. The following theorem states some important examples to demonstrate the applicability of our approach.

\begin{restatable}[Counting Geometric Structures]{theorem}{CountingGeometricStuff}\label{thm:CountingGeometricStructures}
The following non-crossing structures can be counted in $n^{O(\sqrt{n})}$ time on a set of $n$ points in the plane: the set of all graphs, perfect matchings, cycle decompositions, Hamilton cycles, Hamiltonian paths, Euler tours, spanning trees, $d$-regular graphs, and quadrangulations.
\end{restatable}


We would like to emphasize that the proof of
Theorem~\ref{thm:CountingGeometricStructures} uses the algorithm of
Theorem~\ref{thm:AnnotToAlgo} as a black box. Thus these results can be
proved in a completely formal way without the need for revisiting the
details of the proof of Theorem~\ref{thm:AnnotToAlgo}.  
In addition to the actual algorithms presented in the paper, we consider our second
main contribution to be the development of the framework of annotated
triangulations and demonstrating its flexibility in modeling other
problems.

\subsection{The Key Ideas}
%

Our algorithm is based on separators similar to almost all previous algorithms. There are two major differences:  
the separators are defined not in terms of the input (point set), but in terms of the output (triangulation). This gives us higher flexibility so that
we are not restricted to use one kind of separator, but are able to design an algorithm scheme with \emph{two} types of separators.

Given a triangulation with outerplanar index $k>\sqrt{n}$, there exists a cactus layer of size at most $\sqrt{n}$  by the pigeonhole principal, see Figure~\ref{fig:Peeling}. This cactus layer is an ideal candidate for a separator. As cactus layers are nested and separate the inner from the outer part completely. In order to make them canonical, we choose the \emph{outermost} small cactus layer and \emph{peel it off}.
In subsequent recursions, we have to ensure that we count only triangulations, where this was indeed the outermost small cactus layer. This can be done by guessing all possible sizes of all layers that have smaller index. 

\begin{figure}[t]
 \centering
 \includegraphics{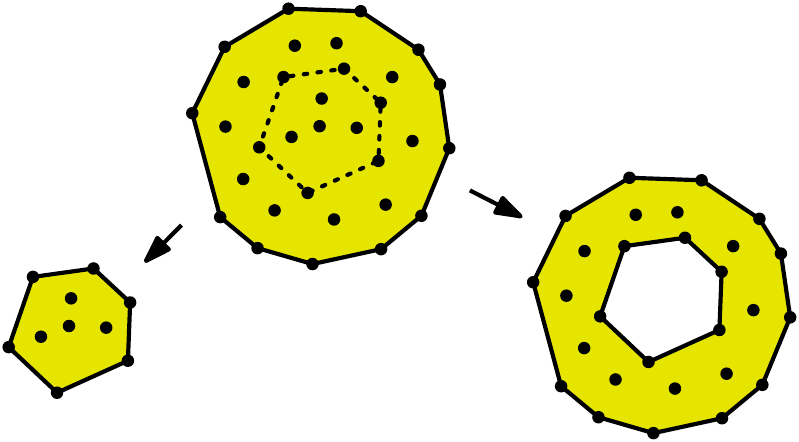}
 \caption{Cactus layes are ideal separators, as they separate in a simple way the inside from the outside and we can easily \emph{peel off} large layers from the outside.}
 \label{fig:Peeling}
\end{figure}

In case that the outerplanar index $k$ is below $\sqrt{n}$, 
 there exists a path of length at most $k-1$ from any vertex to the outermost layer, see Figure~\ref{fig:Nibbling}. (The idea is that every vertex either is adjacent to the outer face or has a neighbor with smaller index.)
If done correctly, these paths can be used as a separator within a well designed dynamic programming scheme. Anagnostou and~Corneil~\cite{anagnostou1993polynomial} demonstrated this for finding a minimum weight triangulation,  using onion layers instead of cactus layers, but the algorithmic idea is essentially the same.
Alvarez, Bringmann, Curticapean and Ray showed how to make these separators canonical and thus suitable for counting problems~\cite{DBLP:journals/dcg/AlvarezBCR15}, by giving each vertex a fixed distinguished rank and always choose the next vertex on the separator-path with the smallest available rank.

\begin{figure}[t]
 \centering
 \includegraphics{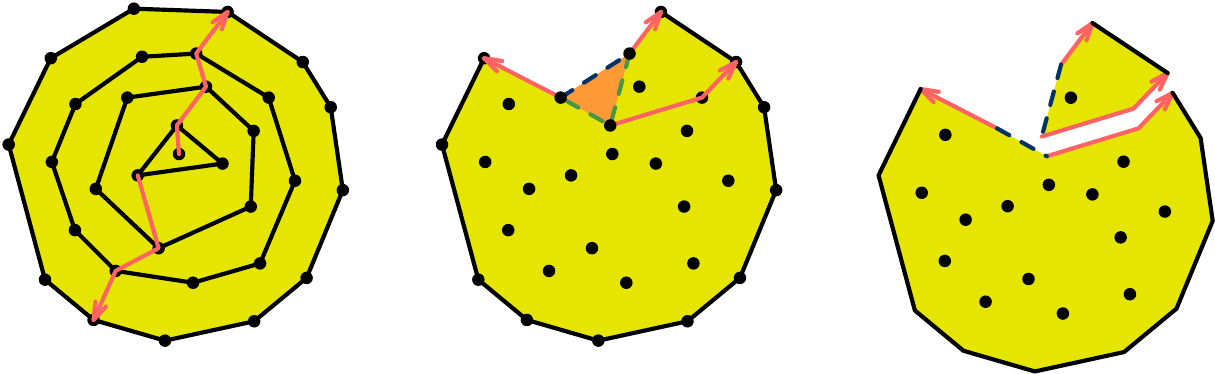}
 \caption{Paths from the interior to the outside are short separators, if the triangulation does not have too many layers.
 The resulting dynamic programming scheme is a little unintuitive, as it \emph{nibbles off} the algorithmic problem from the outside.}
 \label{fig:Nibbling}
\end{figure}

We use these path separators in an, at first, unintuitive way. 
We fix a special edge $e=uv$ on $\partial CH(S)$ and guess all triangles $\Delta = \Delta(u,v,w)$ and a short path $p$ from $w$ to $\partial CH(S)$. 
In this way, we attain a large and a small subproblem. The triangle $\Delta$ defines special edges for subsequent subproblems. We repeat the procedure for all appearing subproblems till the subproblems are of constant size.
We \emph{nibble off} small bites in each recursion step. Only at later stages larger bites are taken.

The technical and conceptual difficulties come from the need to combine the two different separators into \emph{one} dynamic programming scheme. Note that we might first guess a layer then the sizes of layers with smaller index and thereafter use the path separators as described above.
Thus, we have to define very carefully subproblems for our dynamic programing routine that are specifically designed to work for both separators. 

The runtime bound follows from the size of the separators. Whenever we guess a separator of size $O(\sqrt{n})$, we have at most $\binom{n}{O(\sqrt{n})} = n^{O(\sqrt{n})}$ possibilities.
\subsection{Related Work}\label{sec:RelatedWork}

\begin{figure}[tbph]
\centering
{\small
\begin{tabular}[htbp]{|l|l|l|c|}
  \hline
 year &  contribution & idea & cite\\
 \hline
 1979 & polynomial time algorithm for polygons & divide and conquer 
 & \cite{gilbert1979new, klincsek1980minimal}\\
   \hline
 1993 & $n^{O(h)}$ MWT & 
 use short legal paths 
 & \cite{anagnostou1993polynomial} \\
   \hline
 1996 & $O(t(P)\textrm{poly}(n))$ enumeration algorithm & reverse search technique 
 & \cite{avis1996reverse}\\
 \hline
 1999 & $O(tr(P)^2)$ time counting algorithm & D\&C + triangulation paths 
 &\cite{aichholzer1999path} \\
   \hline
 2001 & triangulation path enumeration  & reverse search 
 & \cite{dumitrescu2001enumerating} \\
 & + lower bound & & \\
   \hline
 2004 & empirically fast algorithm & generalize polygon algorithm 
 & \cite{ray2004simple} \\
   \hline
 2012 & $n^{O(h)}$ time algorithm  & legal paths unique using labels 
 & \cite{DBLP:journals/dcg/AlvarezBCR15} \\
   \hline
 2013 & $O(tr(P))$ time algorithm & 
 sweep instead of 
 & \cite{DBLP:journals/corr/AlvarezBR13} \\
 &  triangulation path upper bound & divide and conquer  &  \\
   \hline
2013 & $n^{O(\sqrt{n})}$ time algorithm 
& guess $O(\sqrt{n})$ sized separators
& \cite{DBLP:journals/comgeo/AlvarezBRS15}\\
& with $n^{\left(n^{4/3}\right)}$ approximation-ratio  & solve small subproblmes exact & \\
  \hline
 2013 & $O(2^n n^2)$ time algorithm  & sweep using $x$-monotone curves 
 & \cite{DBLP:conf/compgeom/AlvarezS13} \\
   \hline
 2015 & QPTAS & techniques of Adamszek \& Wiese & \cite{DBLP:conf/icalp/KarpinskiLS15} \\
 \hline
\end{tabular}
}
\caption{Summary of important algorithms to count triangulations.}
\end{figure}

Many authors have worked on finding efficient algorithms for 
counting all triangulations on a given set of points.
Some algorithms have been developed with the purpose of 
finding a minimum weight triangulation, but it turned out
that they can be used or at least adapted to also count
triangulations.

We found it very interesting to observe how ideas used to
compute a minimum weight triangulation can actually be used
to count triangulations instead. 
Here, we survey the relevant work on counting and we also 
try to point out where these ideas have been used before
to compute the minimum weight triangulation.

In 1979, Peter D. Gilbert showed in his master thesis~\cite{gilbert1979new} 
how to compute
a minimum weight triangulation of a simple polygon $P$ using divide
and conquer. In 1980, Klincsek~\cite{klincsek1980minimal} 
showed the same result independently.
The idea of the divide and conquer scheme is based
on the observation that any edge $e$ of any triangulation 
has a triangle $\Delta$ adjacent to $e$. 
The algorithm guesses all potential $\Delta$s 
and recurses on the two arising subproblems.
The number of occuring subproblems is bounded by $n^2$ and
on each a linear number of recursion steps is needed.
Thus the running time is $O(n^3)$.
It is easy to see that this algorithm can be used 
as well to count triangulations. 

Only much later, in 2004, Saurabh Ray and Raymund Seidel showed 
how to use this idea as an algorithm for counting triangulations
for points in the plane~\cite{ray2004simple}. 
While their algorithm seems to be reasonable fast in practice, 
they did not supply a run time analysis.
However, they observed a runtime of $\sqrt{\mathrm{t}(S)}$ empirically,
where $\mathrm{t}(S)$ is the total number of triangulations of $S$, a set of 
$n$ points in the plane.

In 1996, the first enumeration algorithm for counting triangulations 
was published by David Avis and Komei Fukuda~\cite{avis1996reverse}.
They developed the so called \emph{reverse search technique} and applied it
to numerous problems. The runtime of the enumeration algorithm on a set of points $S$ is $t(S)\,n^{O(1)}$.
To see how it works for enumerating triangulations note that the
\emph{flip graph} $G$ of the triangulations is connected.
Let $\Delta_1$ and $\Delta_2$ be two adjacent triangles  
and assume that their union $\square$ is convex. 
Then a \emph{flip} is a replacement of the shared edge by the 
other diagonal of $\square$.
The vertices of the flip graph correspond to the set of triangulations
and two triangulations are adjacent if one can be attained from
the other by a flip.
The strategy of the reverse search technique is to identify
a rooted spanning tree $\tau\subseteq G$ and traverse it.
The root of $\tau$ is the Delaunay triangulation of the underlying point
set. It is known that any sequence of \emph{Lawson-flips} will turn
any triangulation to the Delaunay triangulation. This defines an
acyclic graph with the Delaunay triangulation as unique sink.
To make the outdegree of each triangulation $1$, we give an order
to all potential edges and we always flip the Lawon-edge with
highest priority. This defines $\tau$. In order to traverse the tree
in a depth first manner, the constrained Delaunay triangulation is
needed to keep track of edges, which we do not want to flip.

In 1993 Efthymios Anagnostou and Derek Corneil presented an algorithm 
to compute the minimum weight triangulation of a given point set 
in $O(n^{3h+1})$ time~\cite{anagnostou1993polynomial}. 
Here $h$ denotes the number of \emph{onion layer}s of the given point set.
%
%
The idea of their algorithm is to use \emph{legal paths} from the most inner onion layer to the boundary of the convex hull. The paths are required to visit each layer at most once. These paths can be used as separators.
It is easily seen that there are at most $O(n^h)$ many such paths.
Each subproblem in their dynamic programming approach is defined
by two legal paths and they use one legal path to split their subproblem.
This yields the bound on the running time.
Ketan Dalal \cite{DBLP:journals/rsa/Dalal04} showed in 2004 that
the expected number of onion layers is $\Theta(n^{2/3})$. 
This makes their algorithm subexponential for random point sets on average.

19 years later in 2012 Victor Alvarez, Karl Bringmann, Radu Curticapean and Saurabh Ray
presented an algorithm with the same idea, which could also count the
number of triangulations~\cite{DBLP:journals/dcg/AlvarezBCR15}. The problem of the previous algorithm by
 Anagnostou and Corneil was, that it might potentially over count. 
To prevent this, they need to make sure that for any vertex and 
any triangulation a unique legal path is defined. To do this they gave
every vertex a label and always choose the vertex with smallest label to 
extend their legal path.
With a refined analysis they could show an upper bound of $O(c^n)$ in the 
worst-case, with $c \approx 3.1414 $.

In 1999 Oswin Aichholzer attacked the problem of counting
triangulations from a different angle by introducing 
the concept of a \emph{triangulation paths}~\cite{aichholzer1999path}. 
Given a triangulation $T$ and a line $\ell$ intersecting this triangulation
a triangulation path is a sequence of segments of $T$ 
twining around $\ell$ with some additional technical 
conditions that we do not mention here. 
The remarkable property of the triangulation path is that 
it is unique and thus eligible to be used as 
a potential separator for divide \& conquer. 
The algorithm guesses all potential paths and recurses 
 on all arising subproblems.
We denote by $tr(S)$ the number of all potential 
triangulation paths of $S$ and by $tr(n)= \max_{|P|=n}{tr(P)}$. 
Aichholzer showed that the running time $T(n)$ adheres the recursion 
$T(n) = O(tr(n)T(n/2))$ and thus solves to 
$O(tr(n)^2)$. 
Thus the running time of Aichholzers algorithm
is bounded by $O(tr(n)^2)$.

In 2001, Adrian Dumitrescu, Bernd G{\"a}rtner, Samuele Pedroni and Emo Welzl
reduced the hope that triangulation paths can be used for efficient algorithms~\cite{dumitrescu2001enumerating}.
They provide a simple lower bound example showing $tr(n)\geq 4^{n-\Theta(\log n) }$.
On the positive side they apply the reverse search technique 
to enumerate all triangulation paths efficiently.
They used again the constrained Delaunay triangulations.

Despite hope being diminished, in 2013, Victor Alvarez, 
Karl Bringmann, and Saurabh Ray
presented an algorithm that could count triangulations 
in $O(tr(n))$~\cite{DBLP:journals/corr/AlvarezBR13}. This is a great improvement over Aichholzers algorithm. Their main idea is to use \emph{sweeping} instead of 
divide and conquer. This idea is actually more interesting than 
the running time, as it was later employed for a more efficient algorithm.
Their algorithm is technically non-trivial as it demands
structural insight to triangulation paths in order to identify 
all potential successors of a triangulation path.
They also provide an upper bound of $O(9^n)$ on $tr(n)$.

All these algorithms were subsumed in 2013 by Victor Alvarez and Raimund Seidel~\cite{DBLP:conf/compgeom/AlvarezS13}, winning the best paper award 
on the Symposium of Computational Geometry. 
Their algorithm has a worst-case running time of $O(2^nn^2)$.
Remarkable is their algorithm, because it was 
the first algorithm that provably always counted triangulations
faster than the number of triangulations itself. Further the algorithm
is simple. 
Their idea is to use $x$-monotone curves as separators together
with sweeping. It is easy to see that the number of $x$-monotone
curves on $S$ is bounded by $2^n$. 
Each $x$-monotone curve $c$ defines the subproblem of computing
the number of triangulations below $c$. 
The interaction between subproblems is elegantly encoded
into an algebraic shortest path problem.


In 2013 Victor Alvarez, Karl Bringmann, Saurabh Ray and Raimund Seidel 
presented a simple approximation
algorithm~\cite{DBLP:journals/comgeo/AlvarezBRS15}. 
The running time of their algorithm is $n^{O(\sqrt{n})}$.
It has an approximation factor of $n^{O(n^{4/3})}$.
This is huge but as our aim is to approximate a number
of exponential size it might be not too bad actually.
The idea is to use 
simple cycle separators originally from
G.L. Miller~\cite{DBLP:conf/stoc/Miller84,DBLP:journals/jcss/Miller86} 
and later improved by H.N. Djidjev and 
S.M. Venkatesan~\cite{DBLP:journals/acta/DjidjevV97}.
of size $O(\sqrt{n})$.
Here any triangulation can be counted many times, but
the total over counting can be bounded because after
a certain number of rounds the arising subproblems are 
so small, that they can be solved fast by an exact exponential time algorithm.
In a recursive call of the algorithm all potential 
$n^{O(\sqrt{n})}$ separators are guessed. As there 
is no unique balanced separator for each triangulation 
every triangulation might be counted at most $n^{O(\sqrt{n})}$ times.
To bound the approximation factor, the algorithm computes
the number of triangulations exactly as soon as the subproblems
get small enough.
This algorithm is also able to compute the minimum weight
triangulation and has roughly the same runtime as previous algorithms  
in the worst case~\cite{DBLP:conf/wg/KnauerS06, DBLP:conf/cccg/Lingas98}.

In 2015, an approximation scheme was presented by
Marek Karpinski, Andrzej Lingas, and Dzmitry 
Sledneu~\cite{DBLP:conf/icalp/KarpinskiLS15}.
Their running time is quasi-polynomial that is
$2^{( \log n/\varepsilon )^{O(1)}}$.
The approximation factor is of the form $2^{\varepsilon n}$.
Thus, this algorithm has a better running time, but a
worse approximation factor.
They employ the technique developed by 
Adamszek and Wiese~\cite{DBLP:conf/focs/AdamaszekW13}.
The technique main idea is to use separators, that are very small
and thus all of these separators can be guessed in 
quasi polynomial time. The separators have the disadvantage
that certain solutions are not accounted for.

A similarity of almost all algorithms is that they use 
some kind of separator for a dynamic programming/divide \& conquer scheme.
In 1979, Richard J. Lipton and Robert Endre Tarjan described 
in their seminal paper a simple algorithmic paradigm that
has been employed countless 
times~\cite{lipton1979separator,DBLP:journals/siamcomp/LiptonT80}:
``Three things are necessary for the success and efficiency of
divide-and-conquer: (i) the subproblems must be of the same type as the original and
independent of each other (in a suitable sense); (ii) the cost of solving the original
problem given the solutions to the subproblems must be small; and (iii) the subproblems
must be significantly smaller than the original.''
Let us call this the Lipton-Tarjan-paradigm.
It turns out that (iii) can be replaced by another property.
Namely, that we can bound all potential subproblems that can occur. 
 Denote by $s(n)$ the number of all potential separators.
 If each subproblem is defined by at most
 $c$ separators, the number of potential subproblems is 
 bounded by $s(n)^c$ and the running time is bounded by
 $T(n)\leq s(n)^{c+1}$.
 This strategy requires that each subproblem is defined by at 
 most a constant number of separators.
 It goes particularly well with sweeping as 
 every subproblem is defined by only one separator.

The algorithm
in~\cite{aichholzer1999path,DBLP:journals/comgeo/AlvarezBRS15} follow
the Lipton-Tarjan-paradigm; 
the algorithms in~\cite{DBLP:journals/dcg/AlvarezBCR15, DBLP:journals/corr/AlvarezBR13, DBLP:conf/compgeom/AlvarezS13, anagnostou1993polynomial,
gilbert1979new, klincsek1980minimal} 
follow the second approach.

Regarding the separators used for 
the exact algorithms, we want to mention that
they are designed in a way that each triangulation admits exactly one such separator. To guarantee uniqueness often additional work is required.
To illustrate the importance of uniqueness, recall
that this was the main property of triangulation paths 
proved by Aichholzer~\cite{aichholzer1999path}.
We cannot think of a principal reason, why it should not be
possible that any kind of separator can be made unique.

Most algorithms that are able to count triangulations can also be adopted
to count other kind of geometric structures. 
Studied structures are the set of all potential geometric graphs, 
perfect matchings, spanning trees, cycle partitions, convex partitions and spanning cycles.
This was shown for spanning trees, perfect matchings and spanning cycles by 
Alvarez, Bringmann, Curticapean and Ray in~\cite{DBLP:journals/dcg/AlvarezBCR15} and again by 
Alvarez, Bringmann and Ray~\cite{DBLP:journals/comgeo/AlvarezBRS15}.
In 2014, Manuel Wettstein showed how to adapt the algorithm 
by Alvarez and Seidel to count all the structures mentioned above, see~\cite{WettsteinMasterThesis, DBLP:conf/compgeom/Wettstein14}

The basic strategy is to 
use annotations on the separators. These annotations give 
additional information on the kind of structures to be counted.
The approach of Wettstein uses the specific structure of 
the separators of the algorithm by Alvarez and Seidel,
namely being $x$-monotone curves.
The approach by the other group of authors is more general
and can be applied more broadly.
The first idea is that each geometric structure $\mathcal{X}$ is in one to one 
correspondence to the constrained Delaunay triangulations containing ${\cal X}$.
Thus it is sufficient to count those constrained Delaunay triangulations.
In order to be able to guarantee that two constrained Delaunay triangulations
compose to one, the separators are made 'fat'.

Up to date there exist no lower bounds on the counting problems mentioned so far.
However, as all the afore mentioned algorithms are also able to solve 'decomposable'
problems, we review lower bounds for them.

The most prominent decomposable problem is that of finding a minimum 
weight triangulation. It was a major break-through in 2006,
when Wolfgang Mulzer and G\"{u}nter Rote presented their,
by now famous, NP-hardness proof~\cite{DBLP:journals/jacm/MulzerR08}.
Their proof needed computer assistance, to check that wires and other gadgets
worked as intended. 
The major insight is that the $\beta$-skeleton is always part of a 
minimum weight triangulation, for certain values of $\beta$. 
With this in mind, it is possible to construct 
point sets that enforce certain edges to form 'tunnels' and 'walls'.
Eventually, this becomes complex enough to build suitable
gadgets, to encode any planar positive $1$-In-$3$-SAT formula.

Another, interesting problem is that of computing a triangulation in
case that only a subset of the edges are eligible and the others are forbidden.
It was shown by Errol Lynn Lloyd in 1977 that this problem is NP-hard~\cite{DBLP:conf/focs/Lloyd77}.
This result was strengthened to \#W[2]-hardness, by Alvarez, 
Bringmann, Curticapean and Ray in case that 
the input point set consist of $k$ onion-layers~\cite{DBLP:journals/dcg/AlvarezBCR15}.


\section{Ring Subproblems}

Our algorithm is based on dynamic programming: we define a large number of subproblems that are {\em more general} than the problem we are trying to solve. We generalize the problem by considering {\em rings:} we need to triangulate a point set in a region between a polygon and a cactus. Additionally, we may have layer-constraints prescribing that a certain number of vertices should appear on certain layers.

In this section we present the definition of the ring subproblems
used by the algorithm and show how an algorithm that can solve those
problems implies Theorem~\ref{thm:FullPlaneAlgo} and
Theorem~\ref{thm:AnnotToAlgo} for counting triangulations and
annotated triangulations respectively.  

\begin{figure}[t]
  \centering
  \includegraphics{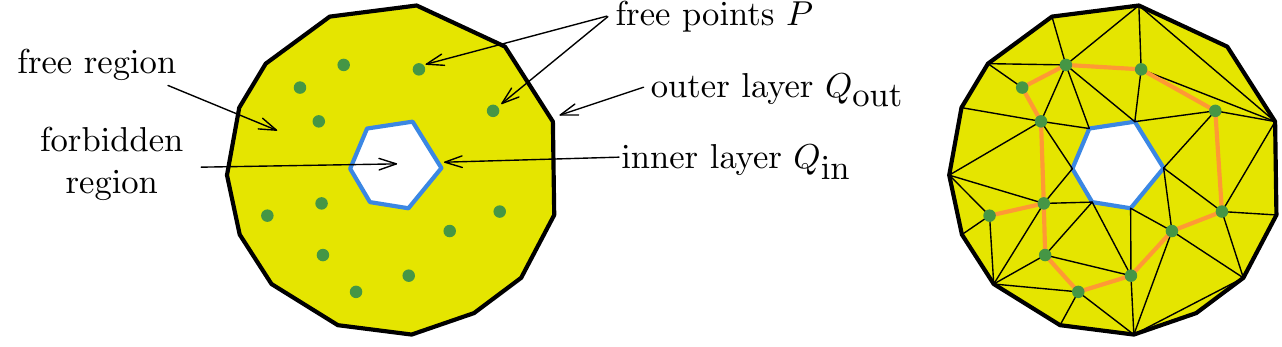}
  \caption{A simple layer-unconstrained ring subproblem and a valid triangulation.
  The width of the ring subproblem is three 
  and the triangulation has three cactus layers.}
  \label{fig:layerSubDef}
\end{figure}

The following definition formally describes ring subproblems, which serve as the basis of the dynamic programming, see Figure~\ref{fig:layerSubDef} for an illustration.

\begin{definition}[Ring Subproblems]\label{def:LSP}
  A \emph{ring subproblem} $\SP$ consists of the following:
  \begin{description}
   \item[outer layer:] The \emph{outer layer} is the non-crossing union of some simple polygons $Q_\textup{out} = Q_\textup{out}(\SP) = (q_1,\ldots,q_a)$. For clarification, $q_1,\ldots,q_a$ are $a$ different simple polygons. And their union forms a cactus graph.
   This is: we require that all edges of $Q_\textup{out}$ are incident to the outer face of $Q_\textup{out}$.
   We will call \SP\ \emph{simple} if $Q_\textup{out}$ consists of only one polygon.

   \item[inner layer:] The \emph{inner layer} $Q_\textup{in} = Q_\textup{in}(\SP)$ is a cactus   contained in $Q_\textup{out}$. 
   We allow the inner layer to be empty. 

 \item[inner/outer layer index and width:] 
   The \emph{inner(outer) layer index} is some positive integer denoted by
   $\ini{\SP}$ (resp., $\out{\SP}$) and associated with the inner (resp., outer) layer.
     The width of a ring subproblem is 
  defined as $w = w(\SP) = \ini{\SP} - \out{\SP} +1 $.
  The meaning of the width is a little tricky. 
  In case that the inner layer is empty it indicates an \emph{upper} bound on the number of layers to be inserted. Otherwise, if the inner layer is non-empty it gives the \emph{precise} number of layers.
 \item[free region:] The region inside $Q_\textup{out}$ excluding the bounded faces of $Q_\textup{in}$ is the \emph{free region}.
 \item[free points:] A set of points $P$ in the interior 
 of the free region. 
  \item[layer-constraints:] 
    This is just a vector $c = (i_1,\ldots,i_n)$ of length $n$,
    with its entries $i_j\in \{0,1,2,3,\ldots, n\}$.
    We refer to $c$ as the \emph{layer-constraint vector} of $\SP$.
    The $j$-th entry of $c$, denoted by $c(j) = i_j$, 
     indicates that the $j$-th layer should have exactly $i_j$ vertices.
    We assume that  $c(j) = 0$, 
    for all $j \notin [\out{\SP}  , \ini{\SP} ]$.
    (Here, $n$ denotes the size of the underlying point set.) 
  \end{description}
  If no layer-constraint vector is specified, we speak of 
  the \emph{layer-unconstrained} ring subproblem. Otherwise, 
  we speak of a \emph{layer-constrained} ring subproblem.
  Given a layer-constraint vector $c$ and a layer-unconstrained 
  ring subproblem \SP , we define $\SP(c)$ as the layer-constrained 
  ring subproblem appended with the layer-constraint vector $c$. 

For the more general algorithm that is also able to count annotated triangulations, additional informations need to be maintained. 
\begin{description}
 \item[boundary annotation:] Some string $s$ on each vertex and edge of the inner and outer layer. 
 \item[annotation system:] 
 An \emph{annotated} triangle is a $9$-tuple consisting of $3$ points, which form an empty triangle and $6$ strings, one for each vertex and edge of the triangle.
  An annotation system is a list $L$ of \emph{annotated} triangles.
\end{description}  
The size $|L|$ of the annotation system $L$ is defined as the total number of annotated triangles.
We assume that the length of each string is in $n^{O(1)}$.
\end{definition}

\begin{definition}[Valid Triangulation]\label{def:ValidTriangRING}
  Given a ring subproblem \SP, consider a graph $T$ extending the 
  graph formed by $Q_\textup{in}\cup Q_\textup{out}\cup P$.
   The graph $T$ can be decomposed into cactus layers $L_i$ as explained in Section~\ref{sec:Prelim}. 
  Here, we slightly change the indexing, by requiring that the first layer $L_j$ has index $j = \out{\SP}$ and the second layer has index $\out{\SP} + 1 $ and so on. 
  We denote by $d(v)$ the index of each vertex defined in this way.
  We call such a graph $T$ of $\SP$ a \emph{valid triangulation 
  of $\SP$} if the following conditions are satisfied:
  \begin{enumerate}[noitemsep,topsep=3pt,parsep=1pt]
  \item  \label{itm:TriangleRING} All faces in the free region 
  are triangles and there are no edges outside the free region.
   \item \label{itm:CondInnerLayerRING} The graph $L_{j}$ with $j = \ini{\SP}$ is the inner layer $Q_\textup{in}$.
   \item \label{itm:ConstraintsRING} All layer-constraints on the layers are satisfied, that is
    $ |V(L_i)|=c(i)$.
  \end{enumerate}
  We drop the last condition in case that \SP \ is a layer-unconstrained ring subproblem, that is,  \SP \ has no layer-constraints.

  Furthermore, a triangulation comes together with an annotations 
  of the edges and the vertices. 
  \begin{enumerate}[noitemsep,topsep=3pt,parsep=1pt]
  \setcounter{enumi}{3}
   \item \label{itm:FeasibleAnnotationsRING} Each empty annotated triangle $\Delta$ of the triangulation $T$ is in the annotation system $L$. 
  \end{enumerate}
\end{definition}

Condition~\ref{itm:CondInnerLayerRING} 
implies that there are exactly $w$ distinct layers in case that 
the inner layer is non-empty, because of the way we defined the indexing.
In case that the inner layer is empty this condition merely 
implies that there are at most $w$ non-empty layers.
   
   We denote by $t(\SP)$ the number of valid triangulations of $\SP$.
 
\begin{theorem}\label{thm:GeoRingAlgo}
  There exists an algorithm that, given an annotated 
  layer-unconstrained ring subproblem \SP\ on $n$ 
  vertices and an annotation system $L$, computes 
  the number of all valid triangulations of \SP\ in 
  $n^{(11+o(1))\sqrt{n}}\cdot |L|^{(12+o(1))\sqrt{n}}$ time.
\end{theorem}

To prove Theorem~\ref{thm:AnnotToAlgo} using Theorem~\ref{thm:GeoRingAlgo}, we start with a ring subproblem where the inner layer is empty and the 
outer layer is the convex hull of the point set. 
There is a technical issue here: the definition 
of the ring subproblem requires a fixed annotation on the outer layer. This means that we can use the algorithm of Theorem~\ref{thm:GeoRingAlgo} only if we try all possible annotations on the boundary, which could be a prohibitively large number of possibilities. Therefore, we use the standard trick of extending the point set to make the convex hull a triangle.

 \begin{proof}[Proof of Theorem~\ref{thm:AnnotToAlgo} by Theorem~\ref{thm:GeoRingAlgo}.]
 We define a ring subproblem $\Or$ such that the triangulations of $S$ and $\Or$ stay in one to one correspondence. 
  Consider  $3$ points forming a triangle $\Delta^+$ containing $S$, see Figure~\ref{fig:TriangleTrick}. The point set $S^+$ is defined as $\Delta^+\cup S$.
\begin{figure}[t]
  \centering
  \includegraphics{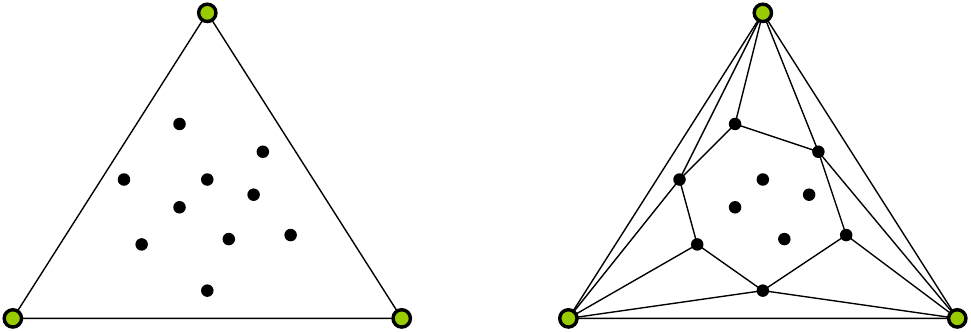}
  \caption{left: The point set $S$ in black and three extra points in green around it. Right: Triangulation $T^*$ of the part outside $CH(S)$.}
  \label{fig:TriangleTrick}
\end{figure}

  We define another annotation system $L^*$ as follows. Fix a triangulation $T^*$ of the region enclosed by $\Delta^+$ and $\partial CH(S)$. 
  Given an edge $e= vw$ of $\partial CH(S)$ and a feasible triangle $\Delta \in L$, with $e \in \Delta$. Let $a_v$, $a_w$ and $a_e$ be the annotation of $v,w$ and $e$ of $\Delta$ respectively. We define $\Delta^*$ as the triangle of $T^*$ incident to $e$ annotated with $a_v$, $a_w$ and $a_e$ for $v,w$ and $e$, and empty strings otherwise.
  To attain $L^*$, we add all triangles of $\Delta^*$ to $L$, as above. 
  Recall that $\mathcal{T}^A({L^*})$ is the set of annotated triangulations with respect to $L^*$. It holds $|\mathcal{T}^A(L)| = |\mathcal{T}^A({L^*})|$.
  With this set up, we are ready to use Theorem~\ref{thm:GeoRingAlgo} by defining an appropriate annotated layer-unconstrained ring subproblem $\Or$. 
%
%
  Then the algorithm to count all annotated triangulations 
  of $S$ is to count all triangulations of $\Or$. 
  This takes $n^{O(\sqrt{n})}\cdot |L^*|^{O(\sqrt{n})}$ time.
  From $|L^*|\leq 2|L|$ follows the running time bound as stated in the theorem.
  We define $\Or$ as follows.
  \begin{description}[noitemsep,topsep=3pt,parsep=0pt,partopsep=0pt]
   \item[outer layer:] $Q_\textup{out} = \Delta^+$ annotated with empty strings.
    \item[inner layer:]  The inner layer $Q_\textup{in}$ is empty.
   \item[free points:]  The {free points} $P$ are exactly $S$.
   \item[inner/outer layer index] $\out{\Or} = 0$ and $\ini{\Or} = n+1$.
    \item[annotation system:] $L^*$
  \end{description}
  Given a set of points $S$ in the plane, it is easy to verify that $\Or$ is a well defined simple annotated layer-unconstrained ring subproblem. 
  Any annotated triangulation $T$ of $S$ can be extended by $T^*$ to a triangulation $T_\textup{ring}$ of $\Or$. 
  We need to check the Conditions~\ref{itm:TriangleRING},~\ref{itm:CondInnerLayerRING} and~\ref{itm:FeasibleAnnotationsRING} 
  of Definition~\ref{def:ValidTriangRING} to confirm that $T_\textup{ring}$ is indeed a \emph{valid} triangulation.
  Condition~\ref{itm:TriangleRING} merely states that $T_\textup{ring}$ is a triangulation.
  Condition~\ref{itm:CondInnerLayerRING} states that the $(n+1)$-st layer of $T_\textup{ring}$ is supposed to be empty. This is true as no triangulation has more than $(n/3 + 1)$ cactus layers. 
  Note that we don't have to check Condition~\ref{itm:ConstraintsRING} as we have \emph{not} specified a layer-constraint.
  For Condition~\ref{itm:FeasibleAnnotationsRING}, we have to show that each empty triangle of $T_\textup{ring}$ belongs to $L^*$. This is the case for the triangles outside $\partial CH(S)$ by the definition of $L^*$ and $T_\textup{ring}$. For each triangle inside $\partial CH(S)$, it follows from the fact that it was true for $T$ and $L$.
  
  It is straight-forward that two different triangulations $T$ and $T'$ of $S$ induce indeed two different triangulations $T_\textup{ring}$ and $T_\textup{ring}'$ of $\Or$, even if $T$ and $T'$ differ only by their annotations.
  
  At last, we need to show that every valid triangulation $T_\textup{ring}$ of $\Or$ comes from some triangulation $T$ of $S$. The triangulation $T$ comes from restricting $T_\textup{ring}$ to $S$.
\end{proof}

  We can easily derive Theorem~\ref{thm:FullPlaneAlgo} from Theorem~\ref{thm:AnnotToAlgo}. 
  
\begin{proof}[Proof of Theorem~\ref{thm:FullPlaneAlgo} by Theorem~\ref{thm:AnnotToAlgo}]
  Choosing $L$ to be the set of all empty triangles of $S$ annotated by empty strings yields the claim. 
  Every triangle of every triangulation of $S$ is contained in $L$ by definition and thus a valid \emph{annotated} triangulation. Given two annotated triangulations $T_1$ and $T_2$. If $T_1$ and $T_2$ are different as annotated triangulations, then they are also different as \emph{plain} triangulations, as every vertex and edge is annotated with the same empty string.
  
  To see that the constant $11$ is achievable, it is necessary to observe that the algorithm presented in this paper also works if it deals without annotation systems. In this case the multiplicative factor $|L|^{(7+o(1))\sqrt{n}}$ is omitted.
\end{proof}

\section{Thin Rings}\label{sec:ThinRings}

This section presents the proof of the following theorem, which gives an algorithm for solving ring subproblems with a certain width $w$. This algorithm will be invoked by the main algorithm for values $w\leq \sqrt{n}$.

\begin{theorem}\label{thm:FullParaAlgo}
  There exists an algorithm that given a simple (layer-constrained or layer-unconstrained) ring subproblem $\SP$ with width $w$ and $n$ free points, computes the number of all valid annotated triangulations of $\SP$ in time $n^{(5+o(1))w} \cdot |L|^{(6+o(1))w}$.
\end{theorem}

Theorem~\ref{thm:FullParaAlgo} implies  Theorem~\ref{thm:ThinPlane} in a similar fashion as Theorem~\ref{thm:GeoRingAlgo} implies Theorem~\ref{thm:AnnotToAlgo}.
%

\begin{proof}[Proof of Theorem~\ref{thm:ThinPlane} via Theorem~\ref{thm:FullParaAlgo}] 
  We describe an algorithm that computes the number $t_{\leq}(S,k)$  of all valid triangulations with outerplanar index at most $k$.
  We can compute the number $t_{=}(S,k)$ of valid triangulations with outerplanar index exactly $k$, by 
  $t_{=}(S,k) \, = \, t_{\leq}(S,k)\, - \, t_{\leq}(S,k-1)$.
  
  We define a layer-unconstrained ring subproblem $\Or$ such that valid triangulation of $\Or$ with outerplanar index $k$ and triangulations of $S$ with outerplanar index $k$ are in one to one correspondence.
  We describe each component of $\Or$ explicitly.  Then we just compute the number of valid triangulations  of $\Or$ using the \textsc{RingSec} algorithm, as stated in Theorem~\ref{thm:FullParaAlgo}. 
  \begin{description}
   \item[outer layer:] The outer layer $Q_\textup{out}$ is the boundary of the convex hull of $S$.   The \emph{outer layer index} equals one.
    \item[inner layer:]  the \emph{inner layer} $Q_\textup{in}$ is empty,
    and its \emph{inner layer index} is $k+1$.
   \item[free points:]  the \emph{free points} $P$ are the points of $S$ not on $\partial CH(S)$.
    \item[boundary annotations:] Each edge and vertex of the inner layer, the outer layer, the boundary paths and the base edge is annotated with the empty string.
    \item[annotation system:] all possible empty triangles annotated with empty strings are feasible.
  \end{description}
    As we have specified all components, it is clear that 
    $\Or$ is indeed a layer-unconstrained nibbled ring subproblem. 
    The one to one correspondence between triangulations of $S$ and $\Or$ follows easily. Let us only emphasize that Condition~\ref{itm:CondInnerLayerRING} of Definition~\ref{def:ValidTriangRING} ensures that any valid triangulation of $\Or$ has \emph{at most} $k$ layers. As we require the $(k+1)$-st layer to be empty.
%
\end{proof}

We use path-separators for the algorithm in Theorem~\ref{thm:FullParaAlgo}. This requires a yet more
specialized definition of subproblems for our dynamic 
programming scheme: \emph{nibbled ring subproblems}.

\begin{figure}[p]
  \centering
  \includegraphics[]{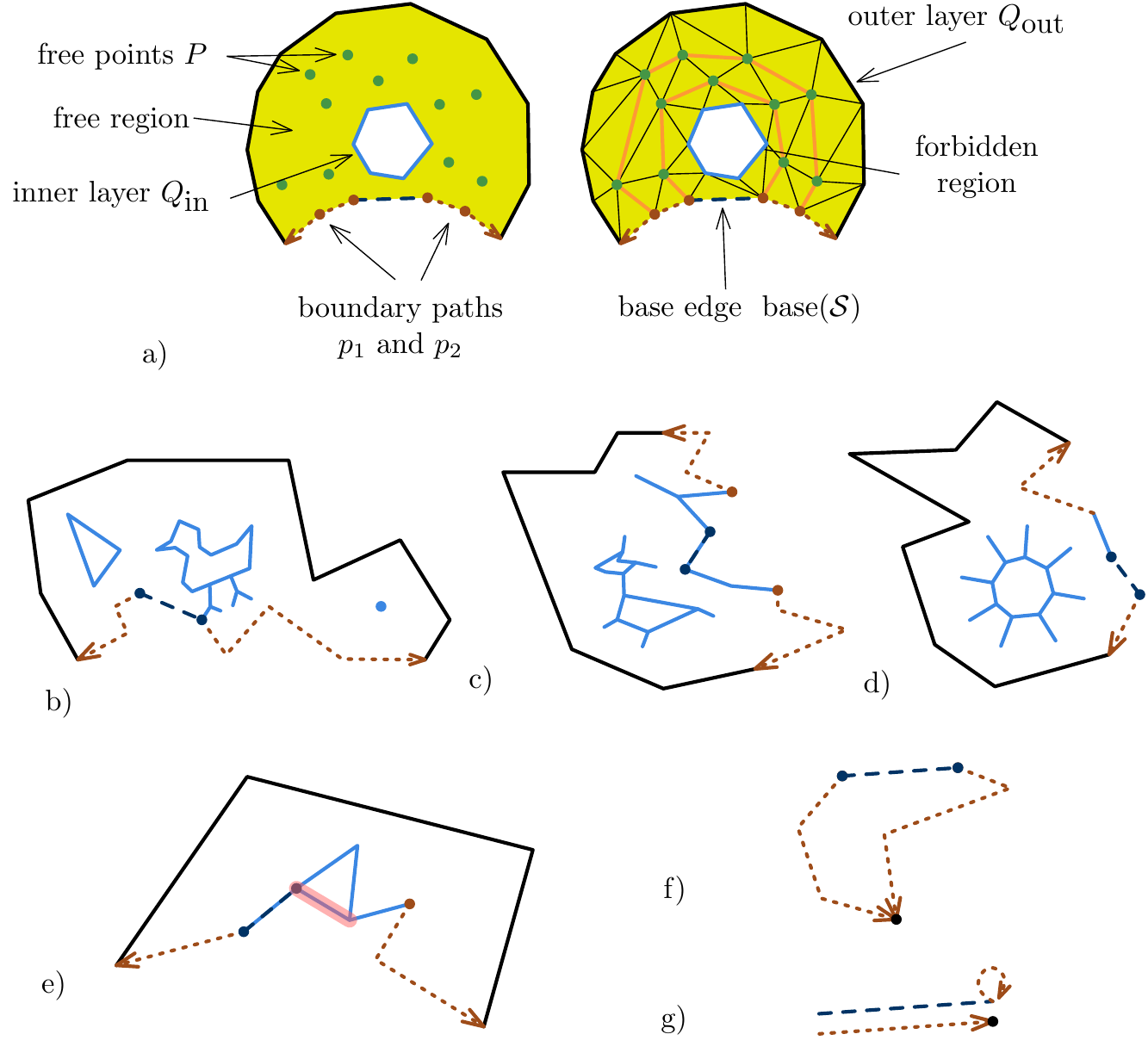}
  \caption{
  a) A nibbled ring subproblem $\SP$ consisting of a base edge $base(\SP)$ depicted in dark blue dashed; 
  two boundary paths, displayed in dotted brown from $base(\SP)$ to the outer layer $Q_\textup{out}$, displayed in black;
  the free region depicted in light yellow;
  containing free points, depicted in green;
  The forbidden region is depicted in white.
  b) The inner layer is a cactus and it may consist of several 
  components.
  It might happen that the base edge shares a vertex with the inner layer.
  c) The special case that the base edge lies on a component $D$ of the 
  inner layer. The boundary paths start, in this situation, 
  from some vertex of $D$. The component $D$ is supposed to have only edges on the free region of \SP .
  d) The special case that a part of the 
  inner layer and the base edge is on the boundary polygon. 
  In this case however, the inner layer and the base edge are edge disjoint.
  e) In this case, the edge marked red is \emph{not} adjacent to the free region 
  and thus this is not a valid nibbled ring subproblem.
  f) The special case that the outer layer consists of a single point.
  g) The special case that \SP\ has an empty free region. 
  The boundary path $p_1$ coincides with the base edge and 
  the outer layer consists of a single point.
  The boundary path $p_2$ consists of a single vertex.
  This case is called degenerate subproblem.}
  \label{fig:SimpleLegalProblem}
\end{figure}

We give a complete self-contained definition of \emph{nibbled ring problems}, see Figure~\ref{fig:SimpleLegalProblem}. As we want the definition to be self contained, we repeat also the parts that are equivalent to ring subproblems. 
These are the width, the inner and outer index, the free region, the free points,  the layer constraints, the boundary annotations and the annotation system.

\begin{definition}[Nibbled Ring Subproblem]
   A \emph{nibbled ring subproblem} $\SP$ can be considered as a plane graph consisting of several components with some extra conditions.
   The edges incident to the outer face form a simple polygon, denoted as \emph{boundary polygon} $Q_\textup{bound}$.
   \begin{description}
   \item[outer layer:] The \emph{outer layer} is a simple non-crossing connected polygonal chain $Q_\textup{out} = Q_\textup{out}(\SP)$ that is a connected subset of the boundary polygon $Q_\textup{bound}$. 
   \item[boundary paths:] Two directed non-crossing paths 
   $p_1(\SP) = u_1\ldots u_b'$ and $p_2(\SP)= v_1\ldots v_b$
   ending at the different end vertices of the outer layer $Q_\textup{out}$. 
   Both paths are part of the boundary polygon and edge-disjoint from the outer layer.
   The length of the two boundary paths must not differ by more than one.
    \item[inner layer:] A cactus $Q_\textup{in} = Q_\textup{in}(\SP)$ called \emph{inner layer}.
    We explicitly allow the inner layer to be empty, but the inner layer index should be always defined. The inner layer may have one or zero components on the boundary, but it is always disjoint from the outer layer and does not share an edge with the boundary paths. To be most precise, the intersection of the inner layer and the boundary must be connected or empty.
   
    \item[inner/outer layer index and width:] 
   The \emph{inner(outer) layer index} is some positive integer denoted by
   $\ini{\SP}$ (resp., $\out{\SP}$) and associated with the inner (resp., outer) layer.
     The width of a ring subproblem is defined as $w = w(\SP) = \ini{\SP} - \out{\SP} +1 $.
  The meaning of the width is a little tricky. 
  In case that the inner layer is empty it indicates an \emph{upper} bound on the number of layers to be inserted. Otherwise, if the inner layer is non-empty, it indicates the \emph{precise} number of layers.

   \item[forbidden regions/free regions:] The bounded faces 
   of the inner layer are called \emph{forbidden region}. 
   The remaining part inside the boundary polygon
   is the free region.
   We demand that all edges of the inner layer are incident to the free region.
   For an example, where this is violated, see Figure~\ref{fig:SimpleLegalProblem}~d) and the edge that is marked red of the inner layer.
   \item[base edge:] One simple segment called \emph{base edge}
   of \SP, denoted by $b=\textup{base}(\SP)$.
   The base edge is always an edge of the boundary polygon
   and edge-disjoint from the outer layer and the boundary  paths. 
   There is a singular exception. The base edge might coincide with
   one boundary path in case the free region is empty, see Figure~\ref{fig:SimpleLegalProblem}~g).
    There are three possible configurations of the base edge together with the inner layer, see Figure~\ref{fig:SimpleLegalProblem}~c)~d) and~e).
   \item[free points:] A set of points $P = P(\SP)$ 
   in the interior of the free region. 
   \item[vertices:] The vertices of a nibbled ring subproblem are all the vertices
   of its components together with the free points.
   \end{description}
   In case, the subproblem $\SP$ is clear from the context, we will suppress it in the notation, that is, we will write $Q_\textup{in}$ instead of $Q_\textup{in}(\SP)$ and so on.
  \begin{description}
    \item[layer-constraints:] 
    This is just a vector $c = (i_1,\ldots,i_n)$ of length $n$,
    with its entries $i_j\in \{0,1,2,3,\ldots, n\}$.
    We refer to $c$ as the \emph{layer-constraint vector} of $\SP$.
    The $j$-th entry of $c$, denoted by $c(j) = i_j$, 
     indicates that the $j$-th layer should have exactly $i_j$ vertices.
    We assume that  $c(j) = 0$, 
    for all $j \notin [\out{\SP}  , \ini{\SP} ]$.
  \end{description}
  If no layer-constraint vector is specified, we speak of the \emph{layer-unconstrained} nibbled ring subproblem.
  We denote by $\SP(c)$ the layer-unconstrained nibbled ring subproblem \SP\  appended with the layer-constraint vector $c$.

For the more general algorithm that is also able to count annotated triangulations, additional informations need to be maintained. 
\begin{description}
 \item[boundary annotation:] Some string $s$ on each vertex and edge of the inner layer, outer layer, the boundary paths and the base edge. 
\end{description}
  We should think of vertex and edge annotations as information that are attached to our subproblem, which is used later to impose some constraints on valid triangulations. It being a string is just a way to encode them.

\begin{description}
\item[annotation system:] 
 An \emph{annotated} triangle is a $9$-tuple consisting of $3$ points, which form an empty triangle and $6$ strings, one for each vertex and edge of the triangle.
  An annotation system is a list $L$ of \emph{annotated} triangles.
\end{description}  
Consider the case that all potential triangles annotated with the empty string form the annotation system $L$ then no actual restriction is imposed.

The size of the annotation system is defined as the length of the list and denoted by $|L|$.
From the list of feasible triangles, we can derive for each vertex and edge a \emph{list of feasible annotations}. It is defined by considering all annotations of feasible triangles with that specific vertex or edge.
\end{definition}

\begin{figure}[t]
  \centering
  \includegraphics{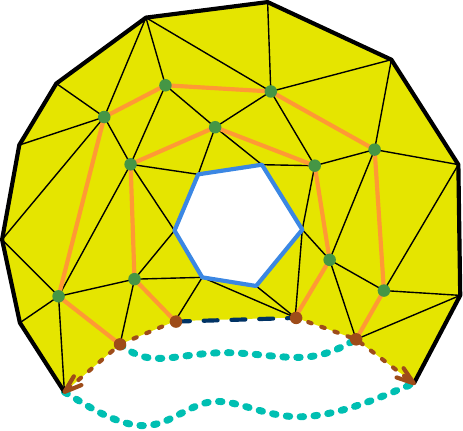}
  \caption{The layers of nibbled ring problems are defined with the help of auxiliary edges (drawn in turquise dotted), which represent the part of the layer of the original triangulation that is already removed (``nibbled'') by previous recursion steps.}
  \label{fig:AuxiliaryEdges}
\end{figure}

\begin{definition}[Valid Triangulation]\label{def:ValidTriangSEC}
 Given a nibbled ring subproblem \SP, consider 
 the boundary paths $p_1(\SP) = u_1\ldots u_{b'}$ and $p_2(\SP)= v_1\ldots v_b$ and assume without loss of generality $b\leq b'$.
 Then, we define $b$ auxiliary abstract edges $(u_{b'},v_b),(u_{b'-1},v_{b-1}),\ldots $, see Figure~\ref{fig:AuxiliaryEdges}. 
 We draw all the edges (not necessarily straight line) in a non-crossing manner so that the outer layer remains incident to the outer face.
 Let $T$ be a triangulation extending the graph formed by 
  $Q_\textup{in}\cup Q_\textup{out}\cup P \cup p_1 \cup p_2 \cup base$.
  We define $T'$ as the triangulation $T$ together with the 
  auxiliary edges, which we just described.
  
  The graph $T'$ can be decomposed into cactus layers $L_i$ as explained in Section~\ref{sec:Prelim}. 
  Here, we slightly change the indexing, by requiring that the first layer $L_j'$ has index $j = \out{\SP}$ and the second layer has index $\out{\SP} + 1 $ and so on. 
  The layers of $T$ are defined by removing the auxiliary edges again.
  We denote by $d(v)$ the index of each vertex defined in this way. It is easy to see that $d(v)$ corresponds to the distance to the outer layer plus the outer-layer index.
  We call such a graph $T$ of $\SP$ a \emph{valid triangulation 
  of $\SP$} if the following conditions are satisfied:
  \begin{enumerate}[noitemsep,topsep=3pt,parsep=3pt,partopsep=0pt]
   \item \label{itm:CondTrianlgesSEC} All faces in the free region are triangles and there are no edges 
   outside the free region.
   \item \label{itm:CondInnerLayerSEC} 
      The graph $L_{j}$ with $j = \ini{\SP}$ is the inner layer $Q_\textup{in}$.
   \item \label{itm:CondConstraintsSEC} All layer-constraints are satisfied, that is
    \[c(i) = |V(L_i)|.\]
  \item \label{itm:CondPathSEC} For any vertex $v_i$ of a boundary path $p$, 
  the successor of $v_i$ ($v_{i+1}$) must satisfy 
   $d(v_i) = d(v_{i+1})+1$.
  Further $v_{i+1}$ must be the neighbor of $v_i$ in $T$ with smallest order label among the neighbors with smaller distance to the outer layer.
  \end{enumerate}
  We will later see that Condition~\ref{itm:CondConstraintsSEC} is always required, as we do not deal with nibbled ring subproblems 
  without layer-constraints.

  Further, a triangulation comes together with annotations of the edges and the vertices. 
  \begin{enumerate}
  \setcounter{enumi}{4}
   \item \label{itm:feasibleTriannglesSEC} Each annotated colored triangle $\Delta$ of the triangulation $T$ is in the annotation system, i.e.\,$\Delta \in L$.
  \end{enumerate}
\end{definition}

We denote by $t(\SP)$ the number of valid triangulations of $\SP$.

\begin{figure}[t]
  \centering
 \includegraphics[width = 0.7\textwidth]{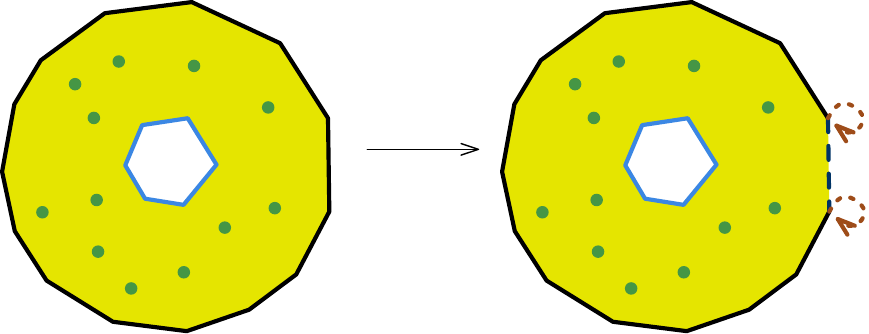}
  \caption{Left: a simple ring subproblem Right: the transformed nibbled ring  subproblem.}
  \label{fig:Transformation}
\end{figure}

In the remainder of this section, we will show the following  result.

\begin{theorem}\label{thm:RingSector}
  There exists an algorithm that
  given a nibbled ring subproblem $\SP$ with width $w$ computes the number of all valid triangulations of $\SP$  in $n^{(5+o(1))w} \cdot |L|^{(6+o(1))w}$ time.
\end{theorem}

Theorem~\ref{thm:RingSector} easily implies Theorem~\ref{thm:FullParaAlgo}, using a simple transformation from ring subproblems to nibbled ring subproblems, see Figure~\ref{fig:Transformation}.
The algorithm indicated in Theorem~\ref{thm:RingSector} is called \textsc{Nibbling} for later reference. The algorithm indicated in Theorem~\ref{thm:FullParaAlgo} is also called \textsc{Nibbling} as
it only transforms the problem into a nibbled ring subproblem.

\begin{proof}[Proof of Theorem~\ref{thm:FullParaAlgo} by Theorem~\ref{thm:RingSector}]
Given a ring subproblem \SP, we define a nibbled ring subproblem \SP' as follows:
\begin{description}
 \item[base edge:] Pick any edge on $Q_\textup{out}(\SP)$ and define it to be the base edge $\textup{base}(\SP')$ of \SP'.
 \item[outer layer:] Let $Q_\textup{out}(\SP')$ be the simple polygonal chain remaining 
 after deletion of $\textup{base}(\SP')$.
 \item[boundary paths:] Are defined to be the endpoints of the base edge and have length zero.
\end{description}
Layer-constraints, indices, the inner layer, free points, annotations and the annotation system carry over without changes, see Figure~\ref{fig:Transformation}.

\vspace{0.2cm}
\noindent \textbf{Claim}(Transformation)
  Let $\SP$ be a simple ring subproblem and $\SP'$ be its transformation as defined above. Then $t(\SP) = t(\SP')$.
\vspace{0.2cm}

We have to show $T$ is a valid triangulation of 
  \SP\ if and only if it is a valid triangulation of $\SP'$.
The free region, the inner layer and the layer-constraints of $\SP$ and $\SP'$ are the same and thus Conditions~\ref{itm:TriangleRING},~\ref{itm:CondInnerLayerRING} and \ref{itm:ConstraintsRING} of Definition~\ref{def:ValidTriangRING} and~\ref{def:ValidTriangSEC} are equivalent.
Condition~\ref{itm:CondPathSEC} of Definition~\ref{def:ValidTriangSEC} is trivially satisfied as both boundary path have length zero.
Also the condition about feasbile triangles does not change.
\end{proof}

%

    \begin{figure}[tp]
  \centering
  \includegraphics{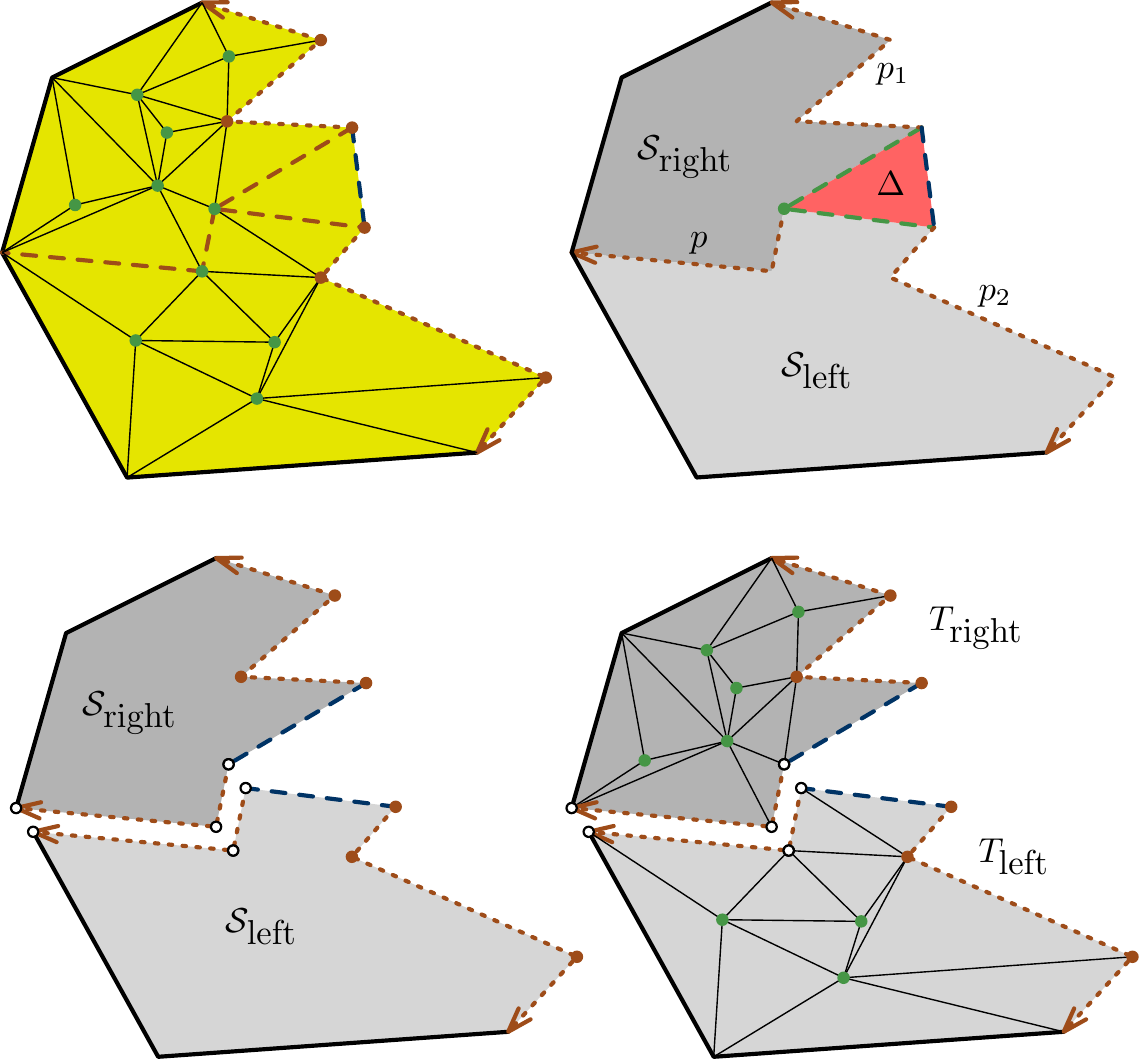}
  \caption{On the top left is a nibbled ring subproblem $\mathcal{S}$ together with a
  valid triangulation $T$. There is a unique triangle $\Delta$ adjacent to the base edges. From the vertex $v$ of $\Delta$ that is not incident to the base edge exists 
  a path $p$ to the outer layer. The triangle $\Delta$ and the path $p$ are drawn dashed brown.
  The path $p$ is uniquely determined, if we always use the vertex with the smallest order label among all available choices. 
  On the top right the nibbled ring problem $\mathcal{S}$ is depicted together with the
  separator path $p$ that splits it into two subproblem $\mathcal{S}_\textup{right}$ and $\mathcal{S}_\textup{left}$.
  At the bottom left, both subproblems  are displayed.
  The three white vertices are shared.
  At the bottom right the restricted triangulations $T_\textup{left}$ and $T_\textup{right}$
  are displayed.}
  \label{fig:legalSplitSoCG}
\end{figure}

\begin{definition}[Degenerate subproblems]
  The subproblems that are not solved recursively are called 
  \emph{degenerate subproblems}, see Figure~\ref{fig:SimpleLegalProblem}~f). 
  To be explicit they consist of
  one base edge, which coincides with one of the boundary paths. 
  The other boundary path has length zero and 
  the outer layer also consists of  a 
single vertex only. Both edges and vertices have some feasible annotation.
  The number of triangulations of an degenerate problem is $1$. 
  It must be checked that each layer has the intended size to see if it is valid.
  (The intended size is given by the layer-constraint vector.)
  All other conditions are always satisfied.
  (Note that this is the only place, where the algorithm actually checks if the constraint condition is satisfied.)
  The number of degenerate subproblems on a set of $n$ points 
  in the plane is quadratic in $n$ and $|L|$.
\end{definition}

For the remainder of this section, we formally define how we split a 
nibbled ring subproblem into two nibbled ring subproblems, see also Figure~\ref{fig:legalSplitSoCG}. As a first step, 
we define separator paths. Their definition is motivated by canonical outgoing paths. Canonical outgoing paths are defined in terms of a
given triangulation $T$. This makes them unsuitable as a separator for 
nibbled ring subproblems, as we do not know the triangulations of our nibbled ring subproblems.
Separator paths are defined only in terms of the nibbled ring subproblem.
However, we need to make sure that every canonical path is also a separator path.
Thereafter, we are ready to define the split of a nibbled ring subproblem $\SP$ into two nibbled ring subproblems 
$(\SP_\textup{left},\SP_\textup{right})$ 
using a separator path $p$. We also have to show how we split the layer-constraints. 
At last we will show that this indeed leads to a recursion 
that counts the number of triangulations correctly.
This boils down to show that a valid triangulation $T$ splits into two valid triangulations  $(T_\textup{left},T_\textup{right})$ and reversely
two valid triangulations $(T_\textup{left},T_\textup{right})$ can be combined to a valid triangulation $T$.

For the following definition, see the top left 
of Figure~\ref{fig:legalSplitSoCG} 
for an illustration.

\begin{definition}[Canonical Outgoing Paths]
\label{def:CanOutPath}
Given a valid triangulation $T$ of a nibbled ring subproblem \SP , we define its \emph{base triangle} $\Delta$ as the unique triangle incident to the base edge. The vertex $v_c$ of the base triangle that is not incident to the base edge is the \emph{base vertex} of $T$.
Then the \emph{\canout path} $p = w_1\ldots w_a$ of $T$ is a directed path in $T$ such that:
\begin{enumerate}[noitemsep,topsep=3pt,parsep=3pt,partopsep=0pt]
  \item \label{Cond:DeltaCOP} The triangle $\Delta$ is empty, feasible and it is formed by $w_1$ and the base edge. 
 \item \label{Cond:ShortPathCOP} The number of vertices  on the path $a$ is bounded by the width $w$. (In short: $a\leq w$.)
 And the last vertex is on the outer layer.
 \item \label{Cond:feasibleCOP} All edges and vertices are annotated with a feasible annotation.
 \item \label{Cond:InnerLayerCOP} Any edge $e$ of $\Delta$ with both endpoints shared with the inner layer must also be on the inner layer.
\end{enumerate}
 There are two more conditions that we separate as they can be only formulated with the help of the underlying triangulation.
 \begin{description}
 \item[outwardness condition:] The distance to the outer layer decreases, when one ``goes''  along the path. More precisely:
 $d(w_{i+1}) = d(w_{i}) -1$, for all $i = 1,\ldots, a-1$.
 \item[order label condition:] For each vertex $w_i$ 
 the vertex $w_{i+1}$ is the neighbor 
 with smallest order label among all 
 neighbors with smaller distance to the 
 outer layer. 
 \end{description}
\end{definition}

Property~\ref{Cond:DeltaCOP} ensures that $\Delta$ is defined as wished.
Property~\ref{Cond:ShortPathCOP} is important for the runtime analysis. It follows implicitly from the outwardness condition 
and the definition of the width.
Property~\ref{Cond:feasibleCOP} is again important for the runtime analysis.
(We only have to guess feasible annotations.)
Property~\ref{Cond:InnerLayerCOP} ensures that the inner layer remains correct.
This is Condition~\ref{itm:CondInnerLayerSEC} of Definition~\ref{def:ValidTriangSEC}.

Every vertex $v$, has some adjacent vertices with lower layer-index according to Lemma~\ref{lem:LayerDistance}. Among all vertices with this property, there is one vertex with smallest order label. Following these edges from the base vertex shows the existence of the \canout path. 
The outwardness condition implies the upper bound on the length of the path.
The order label condition ensures uniqueness.
For later use, we summarize this in the following lemma.
\begin{lemma}\label{cor:UniqueLegalPath}
  Given a valid triangulation $T$ for some nibbled ring subproblem $\SP$. Then there exists exactly one \canout path $p^*(T)$.
\end{lemma}

The outstanding property of canonical outgoing paths is that they \emph{separate} one triangulation into two triangulations in a \emph{canonical} way. 
The technical difficulty is that they are defined in terms of triangulations. 
In order to use them as separators, we have to define a similar notion {\em without} reference to a triangulation. 
This is necessary as our algorithm does not know the triangulations of $\SP$ usually.
These first $4$ conditions are identical to the one of canonical outgoing paths, as they do not rely on some underlying triangulation. 
We replace the outwardness condition and the order label condition appropriately.

\begin{definition}[Separator Path of Nibbled Ring Subproblems]
\label{def:SeparatorPath}
  Let \SP\ be a nibbled ring subproblem, $\Delta$ be a triangle incident to the base edge and $p = w_1\ldots w_a$ a path. We say $\Delta$ and $p$ form a \emph{separator path} of $\SP$ if the following conditions are met:
  \begin{enumerate}[noitemsep,topsep=3pt,parsep=3pt,partopsep=0pt]
  \item \label{itm:TrianglCheck} The triangle $\Delta$ is empty, feasible and it is formed by $w_1$ and the base edge. 
  \item The number of vertices  on the path $a$ is bounded by the width $w$. (In short: $a\leq w$.)
  And the last vertex is on the outer layer.
 \item All edges and vertices are annotated with a feasible annotation.
 \item \label{Cond:InnerLayerSP} Any edge $e$ of $\Delta$ with both endpoints shared with the inner layer must also be on the inner layer.
\end{enumerate}

  Further, we define the \emph{index} of each vertex of $p$ inductively as follows: $\textrm{index}(w_a) = \out{\SP}$ and $\textrm{index}(w_i) =  \textrm{index}(w_{i+1}) +1 $. We define the index in the same manner for the boundary paths. For a vertex $v$ on the outer layer (resp. inner layer) $\textrm{index}(v) = \out{v}$ (resp. $\textrm{index}(v) =\ini{v}$).
  Here, the index  
  plays the role of the distance to the outer layer $d(v)$ as defined in Definition~\ref{def:ValidTriangSEC}.
  We denote by $G$ the graph formed by $\Delta$, $p$ and $\SP$.
  We define $N_G(v)$ as the neighbors of $v$ in $G$.
  \begin{description}[noitemsep,topsep=3pt,parsep=3pt,partopsep=0pt]
  \item[replaced outward condition:] The index can be consistently defined for each vertex on the inner layer, outer layer, the boundary paths and the separator path. In particular for the shared vertices. No two adjacent vertices in $G$ differ in their index by more than one. 
   \item[replaced order label condition:] Let $x_i$ and $x_{i+1}$ be two adjacent vertices  on either one of the boundary paths or the separator paths. 
   Then $x_{i+1}$ has the smallest order label in the set 
    $ \{ \, y \in N_G(v): \textup{index}(x_{i+1}) = \textup{index}(y)\}$.
%
%
  \end{description}
  We denote by $\PP (\SP)$ the set of all possible separator paths of $\SP$. 
\end{definition}

When we define separator paths, we have no access to a triangulation we could refer to. In particular we have no function $d(v)$. We replace it by giving each vertex $v$  an index, which essentially plays the role of $d(v)$.
The replaced outwards condition ensures that the index is at least consistent on $G$. 
We will later see that this is sufficient for the algorithm to work correctly.
Morally, when we never insert an edge that is inconsistent locally, the final triangulation(at the end of the algorithm) is also consistent locally.

\begin{figure}[tp]
  \includegraphics[width = \textwidth]{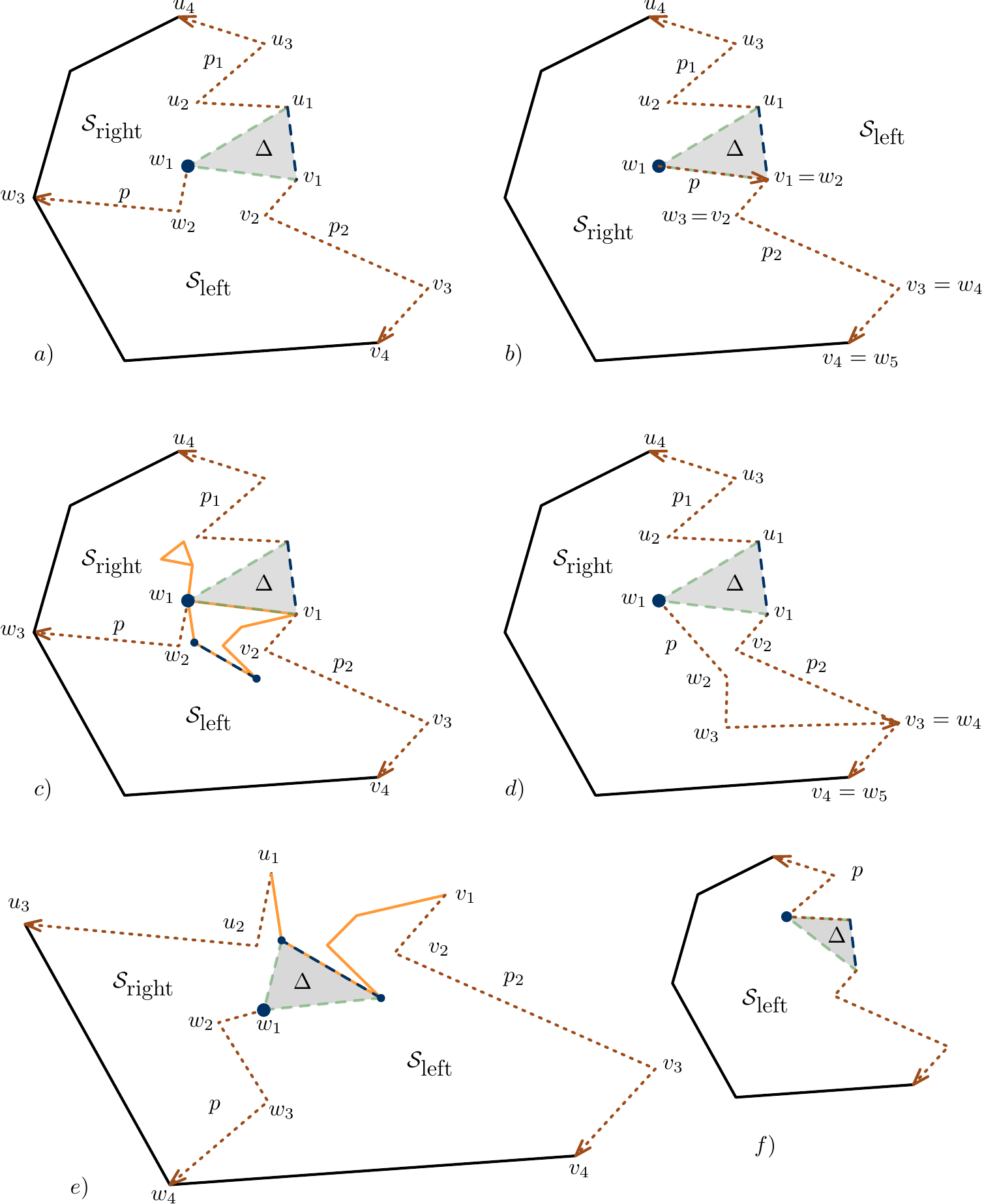}
  \caption{
  a) A split in the common way. i.e., the separator path $p$ uses only free 
points of \SP .  
  b) The separator path $p$ coincides with one of the edges of $\Delta$.
  Thus $\SP_\textup{{left}}$ is an degenerate problem.
  c) One of the edges of $\Delta$ agrees with an edge of $Q_\textup{in}$.
  Therefore we pick any edge of this component of $Q_\textup{in}$ as base edge of 
  $\mathcal{S}_\textup{{left}}$. Also note that we split this component $D$ in this case.
  d) The path $p$ merges at some point with $p_2$. 
  For the subproblem $\SP_\textup{{left}}$ we shorten both boundary paths
  and define their first common vertex as outer layer $Q_\textup{out}$ of $\SP_\textup{{left}}$
  consisting of exactly one vertex.
  e) The subproblems have the base edge adjacent to $Q_\textup{in}$, but not on it.
  f) One side of the triangle $\Delta$ might coincide with one of the boundary paths.
  In this case, the separator path is the separator path shortened by one.}
  \label{fig:legalSplit}
\end{figure}

Note that Condition~\ref{itm:TrianglCheck} is the only place where we actually check if a triangle is feasible. However, this is also the only place where we actually insert triangles.
The important part of this definition is that the set of separator paths include the set of all potential canonical paths.

  The replaced outward condition and the replaced order label condition have some interesting consequences. 
  We will not need to make use of them explicitly in order to show correctness, but we hope to give the reader a better understanding of the replaced order label condition.
\begin{enumerate}[noitemsep,topsep=3pt,parsep=3pt,partopsep=0pt]
 \item The boundary path $p$ has \emph{at most} one vertex on the outer layer.
 \item Assume the separator path $p$ and one of the boundary paths, say $p_1$, share a vertex $v$, then all forthcoming vertices of $p$ and $p_1$ are also shared, see Figure~\ref{fig:legalSplit}~b),~d) and~f).
 \item If the separator path $p$ shares a vertex with the inner layer, 
 it must be the first vertex and $p$ must have length $w-1$ (i.e., $w$ vertices.)
 \item  At last it implies a very technical property. Assume that the start vertex $w_1$ of the separator path and the start vertex $u_1$ of one of the boundary paths are adjacent. Further assume $\textrm{index}(u_1)  = \textrm{index}(w_1)  -1 $. The first consequence is that the separator path is one edge longer than the boundary path. More interestingly, in case that the two paths do not share a vertex, we have that the order label of $w_2$ is smaller than the order label of $u_1$. 
\end{enumerate}

Now, we are ready to show how we split a nibbled ring into two nibbled rings using separator paths.
For the following definitions, see Figure~\ref{fig:legalSplit} for illustrations. Later, we will generate a large number of nibbled ring subproblems arising from them by imposing different layer-constraints on them.

 \begin{definition}[Splitting via Separator Paths]\label{lem:SplitLegalPath}
  Given a path $p= w_1\ldots w_c$ and a triangle $\Delta$ forming a separator path of a nibbled ring  subproblem $\SP$, we can define two layer-unconstrained
  nibbled ring  subproblem $\SP_\textup{left} = \SP_\textup{left}(p)$
  and $\SP_\textup{right} = \SP_\textup{right}(p)$ as follows.
 
  The region bounded by $\SP$ naturally decomposes into $\Delta$, $A_\textup{left}$ and $A_\textup{right}$ as in Figure~\ref{fig:legalSplit}. (Edges can be shared and it is allowed for $A_\textup{left}$ and $A_\textup{right}$ to have no interior points.)
    We specify the components of subproblem $\SP_\textup{right}$ as follows:
    \begin{description}[noitemsep,topsep=3pt,parsep=3pt,partopsep=0pt]
      \item[boundary paths:] Its boundary paths are $p$ and $p_1$.
      If they are not disjoint only include them up to the first vertex, where they agree.
      \item[base edge:] Its base edge $\textup{base}(\SP_\textup{right})$ is usually the edge $e$ of $\Delta$ that is incident to $A_\textup{right}$.
      
      There is a very specific exception. It might be the case that $e$ is incident to a face $f$ of the inner layer. In this case delete $e$ and choose any other edge on $f$, see Figure~\ref{fig:legalSplit}~c) and~e).
     
      \item[outer layer:] Its outer layer $Q_\textup{out}(\SP_\textup{\mbox{right}})$ is the part on $Q_\textup{out}(\SP)$ between
      the endpoints of $p$ and $p_1$. Except if $p$ and $p_1$ have a common vertex,
      then only this vertex forms $Q_\textup{out}(\SP_\textup{\mbox{right}})$. Note this vertex might not be contained $Q_\textup{out}(\SP)$.
      In the first case we define the outer layer index by $\out{\SP_\textup{right}} = \out{\SP}$, in the second case the outer layer index is defined by the index of the first common vertex of $p$ and $p_1$.
      \item[inner layer:] Every connected component of $Q_\textup{in}(\SP)$ 
      that lies in $A_\textup{\mbox{right}}$. It might be that one component $D$
      lies in $A_\textup{\mbox{right}}$ and in $A_\textup{\mbox{left}}$. In this case, we split $D$ in the obvious way. Note that exactly one vertex of $D$ is shared by the two subproblems in this case.
      
      It might happen that one edge is not adjacent to the region $A_\textup{\mbox{right}}$,
      see Figure~\ref{fig:legalSplit}~e).
      In this case, we delete this edge. 
      The inner layer index carries over from \SP .
      \item[free points:] The free points $P(\SP_\textup{right})$ are the 
      points $P(\SP)$ contained in the region $A_\textup{right}$.
      \item[boundary annotations:] Carry over from $\SP$ and the separator path $p$.
      \item[annotation system:] This is exactly the same annotation system as for $\SP$.
      \end{description}

	The subproblem $\SP_\textup{left}$ is defined in the same way, 
	with the replacement of $p_1$ by $p_2$ and $A_\textup{right}$ 
	by $A_\textup{left}$.
    It is easy to check that $\SP_\textup{{left}}$ and $\SP_\textup{{right}}$
    are both layer-unconstrained nibbled ring subproblems.
\end{definition}

      \begin{definition}[Compatible Layer-Constraints]\label{def:PathConstraint}
       Let \SP\  be a nibbled ring subproblem with layer-constraint $c$ and $p \in \PP(\SP)$ a separator path. We define the \emph{layer-constraint vector of $p$}, denoted by $c_p$,
      as follows. Let $p = w_1\ldots w_{a}$ and $p' = w_1\ldots w_{a'}$ be the part of $p$ that is shared by $\SP_\textup{left}$ and $\SP_\textup{right}$, see Figure~\ref{fig:legalSplit}~b),~d) and~f) for examples, where $p \neq p'$. Further let $I$ be the set of indices of the vertices in $p'$. Formally $I$ is defined as $I = \{\, \textrm{index}(v) :  \textup{$v$ is a vertex of $p'$} \,\}$. Recall that the index was defined in  Definition~\ref{def:SeparatorPath}.
      We define $c_p$ as the indicator vector of $I$, to be more explicit: \[c_p(i) =
      \left\{
	\begin{array}{ll}
		1  & \mbox{, if \ }   i \in I  \\
		0 & \mbox{, otherwise.}
	\end{array}
    \right.  \] 
    The purpose of the vector $c_p$ is to ensure that vertices that are shared by the two subproblems are not accounted for twice.
    Further the set of \emph{compatible layer-constraints} $\mathcal{C}(\SP,p)$ is defined by
    \[\mathcal{C}(\SP,p) = \{\, (c_1,c_2) : c_1 + c_2 = c + c_p\mbox{ and } c_i \in \{0,1,\ldots,n\}^n\,\}\]
    \end{definition}

    We denote by $\SP_\textup{\mbox{left}}(p, c_1)$ and $\SP_\textup{\mbox{right}}(p, c_2)$
    the left and right nibbled ring subproblems appended with the layer-constraint vector~$c_1$ and~$c_2$ respectively.

\begin{lemma}[Recursion]\label{lem:CorrectLegalCount}
  Let \SP\ be a nibbled ring subproblem. 
  Then holds
    \[t(\SP) 
   = \sum_{p\in \PP(\SP )} \sum_{(c_1, c_2) \in \mathcal{C}}
   t(\SP_\textup{{left}}(p, c_1))
   \cdot t(\SP_\textup{{right}}(p, c_2)).\]
\end{lemma}

\begin{proof}
``$\leq$'':
    Consider some valid triangulation $T$ of $\SP$ and let $\Delta_c$ be the canonical triangle and $v_c$ the canonical vertex. By Corollary~\ref{cor:UniqueLegalPath} there exists a unique \canout 
    path $p = p^*(T)$ starting at $v_c$. Note that $p$ satisfies Definition~\ref{def:SeparatorPath} and thus can be used as a separator path.
    Thus, it splits $\SP$ geometrically into $\SP_\textup{{left}}$ and $\SP_\textup{{right}}$ as in 
    Definition~\ref{lem:SplitLegalPath}.
    Denote by $T_\textup{{left}}$ and $T_\textup{{right}}$ 
    the triangulations that arise,
    when $T$ is restricted to $\SP_\textup{{left}}$ and $\SP_\textup{{right}}$ respectively.
    Then there exists exactly one pair of layer-constraint vectors $(c_1,c_2)$ with 
    $c_1 + c_2  = c + c_p$
    such that  $T_\textup{{left}}$ and $T_\textup{{right}}$ satisfy $c_1$ and $c_2$ respectively. 
    (Vertices on $p$ are accounted for on $T_\textup{{left}}$ and $T_\textup{{right}}$.)
    It is easy to see that two different valid triangulations $T$ and $T'$ lead two two different pairs of valid triangulations $(T_\textup{{left}},T_\textup{{right}})$ and $(T_\textup{{left}}',T_\textup{{right}}')$
    
    We have to show that $T_\textup{{left}}$ and $T_\textup{{right}}$ are valid triangulations of $\SP_\textup{{left}}$ and $\SP_\textup{{right}}$ respectively.
    It suffices to do this for $T_\textup{{left}}$. We check all conditions of Definition~\ref{def:ValidTriangSEC} explicitly.
    Condition~\ref{itm:CondTrianlgesSEC} follows trivially. 
  Condition~\ref{itm:CondConstraintsSEC} follows from the definition of the layer-constraint vectors $(c_1,c_2)$.
    Condition~\ref{itm:CondPathSEC} follows from the fact the the condition holds for $T$ and the replaced outward condition and the replaced order label condition in the definition of separator paths.
    Condition~\ref{itm:feasibleTriannglesSEC} follows from the fact that every triangle of $T_\textup{left}$ is also a triangle of $T$.

    Condition~\ref{itm:CondInnerLayerSEC} is more intricate.
    We show first by induction that every vertex $v$ in $T_\textup{{left}}$ has the same distance to the outer layer as in $T$.
    This is true for the vertices on the outer layer of $\SP_\textup{left}$.
    Consider the vertex $v$ of $V(T_\textup{{left}})$ with distance $d$ to the outer layer of $\SP$. Then there is a set of other vertices $\{x_1,
    \ldots,x_b \}\subset N_T(v)$ of vertices with smaller distance to the outer layer in $\SP$. 
    At least one of the vertices $x_1, \ldots,x_b$  belongs to $T_\textup{left}$. (Actually, all of them, except $v$ is on the separator path.)
    By the induction hypothesis 
    $x_1, \ldots,x_b$ have  distance $d-1$ to the outer layer in 
    $T$ and in $T_\textup{left}$. Thus $v$ has 
    distance $d$ in $T_\textup{left}$. 
    Thus also all vertices of the 
    inner layer $Q_\textup{in}(\SP_\textup{left})$ 
    have the correct distance to the outer layer.
    
    We also have to show that each edge remains on the layer $L_j$ with $j = \ini{\SP_\textup{left}}$ and no other edge appears on $L_j$.
    We consider first the case that an edge $e = (v,w)$ belongs to $L_j$ of $T_\textup{left}$, but it does not belong to $L_j$ of $T$.
    The only possibility to belong to the $L_j$ in $T_\textup{left}$ is if both endpoints $v$ and $w$ belong to $L_j$ of $T_\textup{left}$
    and thus also $v$ and $w$ belong to the $L_j$ of $T$.
    But there is no edge that has both endpoints on the inner layer of $T$, without also belonging to the inner layer $Q_\textup{in}(\SP)$, as this edge needs to be inside some bounded face of $Q_\textup{in}(\SP)$.
    
    Now, we assume that $e$ belongs to 
    $Q_\textup{in}(\SP)$ and we want to show that $e$ also belongs to the inner layer of $T_\textup{left}$.
    At first note that $v$ and $w$ also belong to 
    $Q_\textup{in}(\SP)$ and this implies that $e$ is not an edge 
    of the boundary path or the separator path.
    Further there exists a triangle $\Delta = \Delta(v,w,x)$ in $T$ 
    incident to $e$ such that $x$ belongs to layer $i-1$. 
    (This condition is necessary and 
    sufficient for an edge to belong to layer $i$.) 
    It is easy to see that this triangle cannot 
    be the base triangle, see Figure~\ref{fig:legalSplit}~c).
    And thus $\Delta$ is also a triangle of 
    $T_\textup{left}$ and this implies $e$ is also an edge of 
    the inner layer of $T_\textup{left}$.
    From this discussion follows Condition~\ref{itm:CondInnerLayerSEC}.
      
    Thus $T_\textup{left}$ and $T_\textup{right}$ are both valid triangulations.

%
%

    ``$\geq$'': Fix any separator path $p\in \PP(\SP)$ and 
    any pair of  compatible constraint vectors $(c_1,c_2)$. 
    Further let $T_\textup{{left}}$ and $T_\textup{{right}}$ be valid triangulations
    of $\SP_\textup{{left}}(p,c_1)$ and $\SP_\textup{{right}}(p,c_2)$ respectively.
    The valid triangulations $T_\textup{{left}}$ and $T_\textup{{right}}$ define a triangulation $T$ by taking the union. We show that $T$ is a {\em valid} triangulation of $\SP$ by going through the complete list of requirements 
    in Definition~\ref{def:ValidTriangSEC}.
    
    First we show that a different pair of valid 
    triangulations $(T_\textup{{left}}', T_\textup{{right}}')$ needs necessarily lead to a different triangulation $T'$. For the sake of contradiction assume that $T = T'$. Then there is a canonical outgoing path $p$ of $T$, which can be used as a separator path. As  $(T_\textup{{left}}', T_\textup{{right}}') \neq (T_\textup{{left}}, T_\textup{{right}})$ at least one pair, say $(T_\textup{{left}}', T_\textup{{right}}')$, comes from $(\SP_\textup{{left}}(q),\SP_\textup{{right}}(q))$ with $q\neq p$. Let $v$ be the last vertex still shared by $p$ and $q$. (As the base triangle is the same there exists at least one such vertex.) And let $w$ be the successor of $v$ on $p$. Without loss of generality $w$ belongs to $T_\textup{right}'$ and thus $T_\textup{right}'$ is not a valid triangulation of $\SP_\textup{{right}}(q)$ 
    as Condition~\ref{itm:CondPathSEC} is violated.
    
    Now let us show that $T$ is a valid triangulation of $\SP$ by checking all conditions explicitly.
    It is easy to see that all faces are triangular and no edge lies outside the free region. Thus Condition~\ref{itm:CondTrianlgesSEC} is satisfied.
    Condition~\ref{itm:CondConstraintsSEC} follows as no vertex changes the distance to the outer layer.
    Condition~\ref{itm:CondPathSEC} follows from the definition of the separator path and the fact that Condition~\ref{itm:CondPathSEC} holds for 
    $T_\textup{left}$ and $T_\textup{right}$.
     Let $p' = v_1,\dots,v_a$ be one of the 
     boundary paths of our subproblem $\SP$.
    We have to show the following technical condition:
    For any vertex $v_i$ of  $p'$ holds that
    the successor $v_{i+1}$
    must be the neighbor of $v_i$ in $T$ with smallest order label among the neighbors with smaller distance to the boundary.
    Note that all neighbors of $v_i$ are either in $T_\textup{left}$ or in $T_\textup{right}$ as $v_i$ is a vertex of a boundary path.
    Thus as Condition~\ref{itm:CondPathSEC} holds for $T_\textup{left}$ and $T_\textup{right}$ it also holds for $T$.
%
    Condition~\ref{itm:feasibleTriannglesSEC} follows from the fact that each empty triangle of $T$ is a feasible  triangle in either $T_\textup{left}$ or $T_\textup{right}$. Further by Condition~\ref{itm:TrianglCheck} of Definition~\ref{def:SeparatorPath} (separator path) the base triangle is also feasible. 
    
    For Condition~\ref{itm:CondInnerLayerSEC} of Definition~\ref{def:ValidTriangSEC}, 
    we have to show that $Q_\textup{in}(\SP)$ 
    is the layer $L_j$ of $T$ with $j = \ini{\SP}$.
    Note that the distance to the boundary of each vertex of $T_\textup{left}$ and $T_\textup{right}$ is the same as in $T$. 
    Every vertex $v$ shared by $T_\textup{left}$ 
    and $T_\textup{right}$ has the same distance 
    to the boundary $d(v)$ due to the enforced 
    distances by the indices on the  separator path $p$. 
    This implies the vertices 
    of $Q_\textup{in}(\SP)$ correspond to the vertices of $L_j$.
    It remains to show that the edges are the same as well.
    For this purpose let $e$ be an edge of $Q_\textup{in}(\SP)$ then it is 
    either an edge of $T_\textup{left}$, $T_\textup{right}$ 
    or an edge of the base triangle of $\SP$. 
    Thus it must also be an edge of $L_j$ by the fact that $T_\textup{left}$ and  $T_\textup{right}$ satisfy Condition~\ref{itm:CondInnerLayerSEC} 
    of Definition~\ref{def:ValidTriangSEC} and Condition~\ref{Cond:InnerLayerSP} of Definition~\ref{def:SeparatorPath}.
    The reverse direction goes by the same argument.
    This shows that $T$ satisfies all conditions of Definition~\ref{def:ValidTriangSEC} and thus is a valid triangulation.
\end{proof}

\begin{algorithm}
  \caption{\textsc{Nibbling}: Counting triangulations of nibbled rings with width $w$.}
  \label{alg:LegalCounting}
  \begin{algorithmic}[1]
  \small
    \Require{$\SP$ to be a nibbled ring subproblem of width $w$;
    }
    \State Initialize Search tree $\tau$ to store all 
    computed values and insert the degenerate subproblems.
    \State \Return{\textsc{Count}(\SP)}
    
    \Statex
    \Function{Count}{\SP}
    \If{\SP\ in $\tau$}
	\State \Return{$\tau(\SP)$}
      \EndIf
      \State $s=0$
    \For{separator paths $p \in \PP(\SP)$}
     \For{ constraints $c_1 + c_2 = c + c_p$}
	  \State Define $\SP_\textup{\textrm{left}}(p,c_1)$ and $\SP_\textup{\textrm{right}}(p,c_2)$ 
	    using $p$ and
	    constraints $(c_1,c_2)$.
	  \State  $s = s \ +\ $\textsc{Count}($\SP_\textup{\textrm{left}}(p, c_1)$) 
	  $\cdot$ \textsc{Count}($\SP_\textup{\textrm{right}}(p, c_2)$)
      \EndFor
    \EndFor
      \State insert $(\SP,s)$ into $\tau$
      \State \Return{$s$}
    \EndFunction
    
  \end{algorithmic}
\end{algorithm}

We are now ready for the proof of Theorem~\ref{thm:RingSector}.
\begin{proof}[Proof Theorem~\ref{thm:RingSector}]
  We will first describe the \textsc{Nibbling} algorithm, then show its correctness and finally supply a runtime analysis.

   The Algorithm \textsc{Nibbling} uses the memoization technique and is based on dynamic programming. The subproblems of the dynamic programming scheme are the nibbled ring subproblems. Already computed solutions are stored in a search tree denoted by $\tau$.
   At the beginning all degenerate subproblems are inserted into $\tau$ with its correct value.  Otherwise each nibbled ring subproblem is solved recursively, by the recursion of Lemma~\ref{lem:CorrectLegalCount}. The pseudocode of \textsc{Nibbling} is depicted as Algorithm~\ref{alg:LegalCounting}. 
   
  It is easy to see that there exists exactly one triangulation for each degenerate problem. In order to check if this triangulation is valid, we have to check only if the layer-constraint are satisfied. All other conditions of Definition~\ref{def:ValidTriangSEC} are trivially satisfied.
  And thus the algorithm will return the correct output for these subproblems. By induction, all other subproblems are computed correctly as well. The induction step is done in Lemma~\ref{lem:CorrectLegalCount}.
 
  The algorithm runs in $n^{O(w)}\cdot |L|^{O(w)}$ time. To see this, we give an upper bound on the total number of subproblems and the total number of recursive calls that can potentially appear. We also have to account for searches in the searchtree $\tau$.
  But each search has costs of $\log ( n^{O(w)}) = w \log n = O(n^2)$. We add these costs to the recursive calls.
  
  Denote by \SP\ the initial input. Then all subsequent subproblems 
  appearing are defined by the two boundary paths together with the 
  base edge and some constraints. 
  In the case that the outer layer of some subsequent subproblem does not coincide with the outer layer of the initial nibbled ring subproblem appears only if the two boundary paths share their last vertex. Thus also in this case the outer layer is completely defined by the boundary paths.

  Note that every boundary path consists of 
  at most $w+1$ vertices and $w$ edges. 
  And thus there are at most 
  $n^{w+O(1)}\cdot |L|^{2w+O(1)}$ 
  possible annotated boundary paths possible.
  The number of possible constraint vectors 
  is bounded by $n^{w}$, as there are 
  at most $w$ entries that are not predetermined 
  to be zero, the size of the outer layer or 
  the size of the inner layer. 
  Thus the total number of potential 
  nibbled ring subproblems is 
  bounded by $n^{3w+O(1)}\cdot |L|^{4w+O(1)}$
  
  The number of recursive calls per ring subproblem is bounded by the number of separator paths times the number of pairs of compatible constraint vectors $(c_1,c_2)$. 
  Given $c_1$, there exists only one constraint vector $c_2$ compatible to it. Thus there are at most $n^{w}$ compatible pairs. 
  
  The number of separator paths is bounded by by $n^{w+O(1)}\cdot |L|^{2w+O(1)}$ in the same way as we bounded the number of boundary paths.
  Thus the number of recursive calls is bounded by $n^{2w+O(1)}\cdot |L|^{2w+O(1)}$. This remains true even if we add the costs for the searches in $\tau$.
  
  The total running time is thus $O(n^{5w+O(1)}\cdot |L|^{6w+O(1)}) = n^{O(w)}\cdot |L|^{O(w)}$ as claimed.
\end{proof}




\section{General Layer-Unconstrained Ring Subproblems}\label{sec:GeneralRing}

Here, we give describe the algorithm to 
count Layer-Unconstrained Ring Subproblems.
We start, to define formally the layers, 
we are aiming to use as separators and 
show first that every triangulation has 
exactly one. Then we define a set of all 
separators, we want to recurse on. 
Further, we show how to split a ring subproblem 
using a layer separator into an inner 
and outer ring subproblem.
For the outer ring subproblem, we will define appropriate 
layer-constraints, to ensure that the layer 
separator we used is indeed canonical for 
all triangulations, we will count henceforth. 
We finish with a full description of the algorithm, 
a proof of correctness and an upper bound on the running time.

 \begin{definition}[Peripheral Layers]
  Given a valid triangulation $T$ of a 
  ring subproblem \SP, we say
  Layer $L_i$ of $T$ is 
  \emph{$m$-peripheral} if $\out{\SP}  < i \leq \out{\SP} + m $.
  Lemma~\ref{lem:LayerProbs1} and the definitions 
  following thereafter depend on a parameter $m$, 
  which will  be chosen later. 
  \end{definition}
  
%

\begin{lemma}[Layer Separation]\label{lem:LayerProbs1}
  Let $\SP$ be a ring subproblem of width $w = w(\SP)\geq m +2 $ 
  with at most $n$ free points. And let $T$ be a valid triangulation of $\SP$.
  Then there exists
  \begin{enumerate}[label=(\alph*), noitemsep,topsep=3pt,parsep=3pt,partopsep=0pt]
   \item \label{itm:ExistLayer} 
    an $m$-peripheral layer $L$ of $T$ of size $\leq \left\lfloor\frac{n}{m}\right\rfloor$,
   \item \label{itm:uniqueLayer}
   the smallest index $i^*$ with $|V(L_{i^*})|\leq \left\lfloor\frac{n}{m}\right\rfloor$ of $T$ is unique and
   \item \label{itm:LayerSepar}
   the layer $L_{i^*}$ separates the inner layer from the outer layer.
   This implies every connected component of $Q_\textup{in}(\SP)$ is in some bounded face
   of $L_{i^*}$.
   \item \label{itm:layerSizes} All $m$-peripheral layers further outside of $L_{i^*}$ have size at least $\left\lfloor\frac{n}{m}\right\rfloor +1$.
   \item \label{itm:LayerAnnot} All edges and vertices of $L_{i^*}$ have a feasible annotation.
  \end{enumerate}
\end{lemma}
\begin{proof}
    \ref{itm:ExistLayer} There are exactly $m$ potential layers, each pair of layers is vertex disjoint and there are at most $n$ potential vertices. The claim follows by the pigeonhole-principle.
   \ref{itm:uniqueLayer} 
    This is immediate from the definition.
  \ref{itm:LayerSepar}
    This follows from the way we defined our layers.
  \ref{itm:layerSizes}
  This follows from the definition of $L_{i^*}$.
  \ref{itm:LayerAnnot} follows from the fact that every edge of a valid triangulation has some feasible annotation.
\end{proof}


\begin{definition}[Peripheral Layer Separators] \label{def:LayerSeparator}
Given a  ring subproblem $\SP$, we define 
a \emph{peripheral layered separator} $L$ of $\SP$ 
as a cactus with an index, denoted by $\ind{L}$, such that
\begin{enumerate}[noitemsep,topsep=3pt,parsep=5pt,partopsep=0pt]
  \item $L$ does not induce any crossings with $\SP$ and $V(L) \subseteq P(\SP)$.
  \item Every connected component of 
    $Q_\textup{in}(\SP)$ is inside some bounded face of $L$
  and $L$ is contained in the interior 
  of the outer layer $Q_\textup{out}(\SP)$.
  \item The index of $L$ satisfies $\out{\SP} < \ind{L} \leq \out{\SP} +m $.
  \item $|V(L)|\leq \left\lfloor\frac{n}{m}\right\rfloor$.
  \item All edges and vertices of $L$ have a feasible annotation.
\end{enumerate}
We denote with \LL(\SP ) the set of all peripheral layer separators of \SP .
\end{definition}

In Definition~\ref{def:SplitLayer}, we will define an inner and outer layer-unconstrained ring subproblem $\SP_\textup{in}$ and $\SP_\textup{out}$, given a simple layer-unconstrained ring subproblem $\SP$ and a peripheral separator layer $L$. Thereafter in Definition~\ref{def:OuterConstraints}, we will define a set of constraints $\mathcal{C}_\textup{out}$ for the layer-unconstrained ring subproblem $\SP_\textup{out}$.
At last Lemma~\ref{lem:CorrectCountSplitByLayers} describes a 
recursion and asserts that this recursion can be used to correctly count all valid triangulations of \SP . 

\begin{figure}
 \centering
 \includegraphics{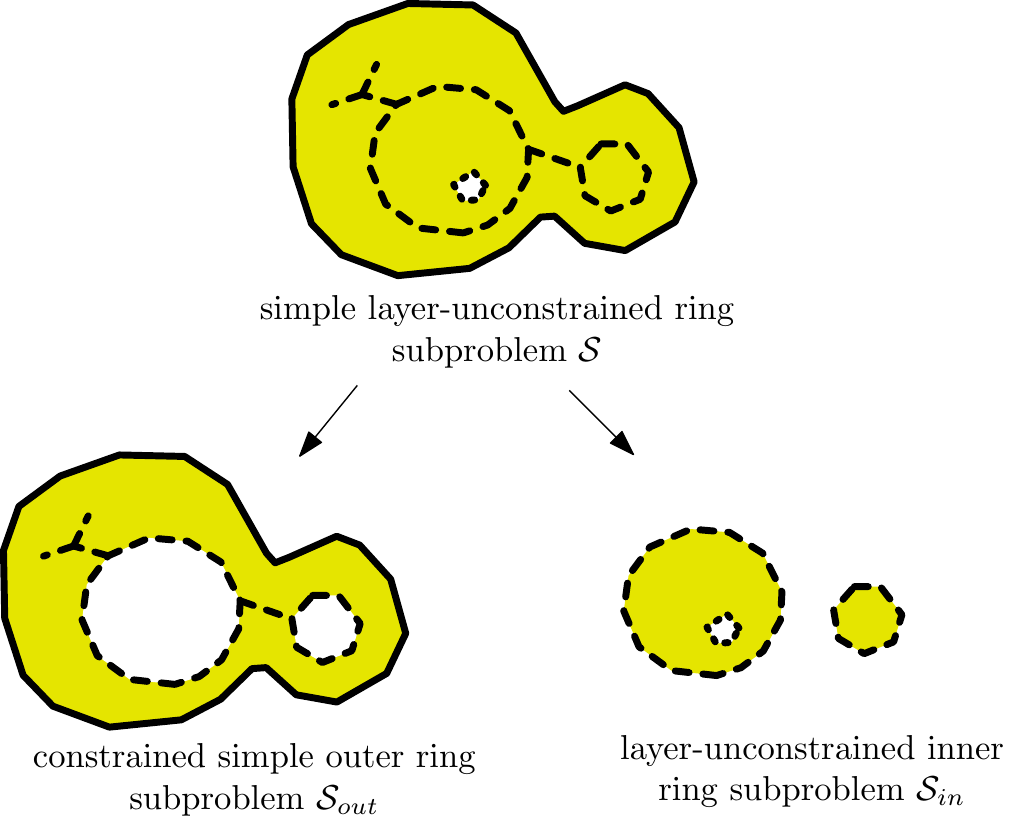}
 \caption{Illustration how to \emph{peel off} the outer ring subproblem.}
\end{figure}

\begin{definition}[Split by Layers]\label{def:SplitLayer}
Given a \emph{layer-unconstrained, simple}, ring subproblem \SP\ together with a peripheral layered separator $L \in \LL(\SP)$, we are now ready to describe the two arising layer-unconstrained ring subproblems. One of them is called the \emph{outer problem} and the other is called the \emph{inner problem}.
We define the layer-unconstrained outer problem $\SP_\textup{out} = \SP_\textup{out}(L)$ as follows:
\begin{description}[noitemsep,topsep=3pt,parsep=3pt,partopsep=0pt]
 \item[outer layer:] $Q_\textup{out}(\SP_\textup{out}) = Q_\textup{out}(\SP)$.
 \item[inner layer:] $Q_\textup{in}(\SP_\textup{out}) = L$.
 \item[outer/inner layer index:]$\out{\SP_\textup{out}} = \out{\SP}$ and $\ini{\SP_\textup{out}} = \ind{L}$.
 \item[free points] $P(\SP_\textup{out}) $ is the subset of $P(\SP)$ in the outer face of $L$.
 \item[boundary annotations:] Carry over from $\SP$ and the separator layered separator $L$.
 \item[annotation system:] This is exactly the same annotation system as for $\SP$.
\end{description}
We do not specify the free region, vertices and the width as they arise from the components given above. 
It might have seemed at first a little unmotivated that we allowed $Q_\textup{out}(\SP)$ 
to consists of more than one polygon, but now we need it as our layers might have more than one component. We define the layer-unconstrained inner problem $\SP_\textup{in}= \SP_\textup{in}(L)$ as follows:
\begin{description}[noitemsep,topsep=3pt,parsep=3pt,partopsep=0pt]
 \item[outer layer:] The outer layer $Q_\textup{out}(\SP_\textup{in})$ is defined by the simple polygons defined by the bounded faces of $L$. The outer layer index $\out{\SP_\textup{in}} = \ind{L}$.
 \item[inner layer:] $Q_\textup{in}(\SP_\textup{in}) = Q_\textup{in}(\SP)$, with $\ini{\SP_\textup{in}} = \ini{\SP}$.
 \item[free points:] $P(\SP_\textup{in})$ is the subset of $P(\SP)$ in 
 the bounded faces of $L$. 
 \item[boundary annotations:] Carry over from $\SP$ and the separator layered separator $L$.
 \item[annotation system:] This is exactly the same annotation system as for $\SP$.
\end{description}
\end{definition}
Note that if $Q_\textup{in}(\SP)$ is 
empty then $Q_\textup{in}(\SP_\textup{in})$ will be empty as well.
We only want to count triangulations $T$ in $\SP_\textup{out}$ that have all their layers larger than $ \left\lfloor\frac{n}{m}\right\rfloor$. 
\begin{definition}\label{def:OuterConstraints}
  Given a simple layer-unconstrained ring subproblem \SP\ and an peripheral layer separator $L \in \LL(\SP)$, we define a set of outer constraints 
  $\mathcal{ C}_\textup{out}$ as follows. We say 
  $(i_1,\ldots, i_n) \in \mathcal{ C}_\textup{out}$ if the following conditions on $i_j$ are satisfied for all $j$. 
  \begin{enumerate}
   \item $i_j = |Q_\textup{out}(\SP)|$, in case that $j = \out{\SP}$.
   \item $i_j = |V(L)|$, in case that $j = \ind{L}$.
   \item $i_j = 0$, in case that  $j \notin [\out{\SP},\ind{L}] $.
   \item $i_j \geq \left\lfloor\frac{n}{m}\right\rfloor $, if $j \in [\out{\SP} + 1,\ind{L} -1]$.
  \end{enumerate}

\end{definition}

The following lemma describes a recursion and asserts its correctness.

\begin{lemma}
[Correctness $\textsc{SplitByLayers}$]
\label{lem:CorrectCountSplitByLayers}
Let \SP\ be a simple layer-unconstrained ring subproblem and $L$ some separation layer of \SP .
Then $\SP_\textup{out}$  and $\SP_\textup{in}$ are  layer-unconstrained ring subproblems, in particular $\SP_\textup{out}$ is again simple.
It holds
   \[t(\SP) = \sum_{L\in\LL(\SP)}t(\SP_\textup{in}(L))\cdot \left(
   \sum_{c\in\mathcal{C}_\textup{out}}  t(\SP_\textup{out}(L,c))\right).\]
\end{lemma}
\begin{proof} By definition, $\SP_\textup{out}$  and $\SP_\textup{in}$ are layer-unconstrained ring subproblems and it holds that $\SP_\textup{out}$ is simple as it has the same outer layer as \SP . It remains to show the recursion.

 ``$\leq$''  \ Given a triangulation $T$ of \SP, by Lemma~\ref{lem:LayerProbs1}, the graph $L_{i^*}$ satisfies all conditions of Definition~\ref{lem:LayerProbs1} and we can split $\SP$ as described into two layer-unconstrained subproblems $\SP_\textup{in}$ and
 $\SP_\textup{out}$.
  The triangulation $T$ decomposes naturally into two valid triangulations $T_\textup{in}$ and  $T_\textup{out}$. In particular there exists exactly one $c \in \mathcal{C}_\textup{out}$ such that $T_\textup{out}$ is a valid triangulation of $\SP_\textup{out}(c)$. 
  We check explicitly all conditions of Definition~\ref{def:ValidTriangRING} to show that $T_\textup{out}$ and $T_\textup{in}$ are indeed valid triangulations.
  Recall that Condition~\ref{itm:TriangleRING} only asks that a valid triangulation has no crossings or edges outside the free region.
  Thus Condition~\ref{itm:TriangleRING} is trivially true by the way the triangulations are defined.
  Recall that Condition~\ref{itm:CondInnerLayerRING} asks for the inner layer
  of the subproblem to coincide with the most inner layer of the triangulation with the correct index.
  As $L_{i^{*}}$ is defined to be the first small layer. 
  It is in particular the inner layer of $\SP_\textup{out}$ with index $i^*$. 
  In other words Condition~\ref{itm:CondInnerLayerRING} 
  is satisfied for $T_\textup{out}$. 
  By the way we defined cactus layers, the layer structure of $T_\textup{in}$
  is inherited by $T$ and Condition~\ref{itm:CondInnerLayerRING}  is also satisfied for $T_\textup{in}$.
  Recall Condition~\ref{itm:ConstraintsRING} asks that each layer has the size given by the layer-constraint vector. We have defined the layer-constraint vector $c$ to do match the size of each layer of $T$. Thus also this condition is satisfied.
  At last Condition~\ref{itm:FeasibleAnnotationsRING} requires that each triangle is feasible. This follows for $T_\textup{out}$ and $T_\textup{in}$ from the fact that all their empty triangles are triangle of $T$ and the fact that $T$ satisfies Condition~\ref{itm:FeasibleAnnotationsRING}.
   Thus $T_\textup{out}$ and $T_\textup{in}$ are valid triangulations.
   
   Also note that a different triangulation $T'\neq T$ would lead to a different pair of triangulations 
   $(T_\textup{out}',T_\textup{in}') \neq (T_\textup{out},T_\textup{in})$.

  ``$\geq$'' \  Assume we have given some separator layer $L \in \LL(\SP)$ and some constraint $c\in \mathcal{C}_\textup{out}$.
  Further, let $T_\textup{in}$ be a valid triangulation of $\SP_\textup{in}$ and $T_\textup{out}$ some
  valid triangulation of $\SP_\textup{out}(c)$. Then the union of $T_\textup{out}$ and $T_\textup{in}$ forms a valid triangulation $T$ of $\SP$. As \SP\ has no constraints there are only three conditions to be checked.
  
  Condition~\ref{itm:TriangleRING} is again trivially true as $T_\textup{out}$ and $T_\textup{in}$ are nested and thus cannot cross each other.
  Condition~\ref{itm:CondInnerLayerRING} follows from the fact that 
  the layers of $T$ are exactly the layers of $T_\textup{out}$ and $T_\textup{in}$. Thus Condition~\ref{itm:CondInnerLayerRING} follows for 
  $T$ because it was true for $T_\textup{in}$ and the fact that the inner 
  layer of $\SP_\textup{in}$ and $\SP$ coincide.
  We don't need to satisfy Condition~\ref{itm:ConstraintsRING} as $\SP$
  is layer-unconstrained.
  Condition~\ref{itm:FeasibleAnnotationsRING} follows again 
  from the fact each triangle of $T$ is a triangle of either 
  $T_\textup{out}$ or $T_\textup{in}$.
  
  Also note that a different pair of triangulations  
  $(T_\textup{out}',T_\textup{in}') \neq (T_\textup{out},T_\textup{in})$
  would lead to a different 
   triangulation $T'\neq T$.
\end{proof}

\begin{remark}
  We define a \emph{disk subproblem} as a ring subproblem with an empty inner layer. Coincidentally, whenever the \textsc{SplitByLayers} procedure is applied, it is applied to a disk problem rather than an layer-unconstrained ring subproblem. So big parts of the algorithm do not need ring subproblems in full generality. We could nevertheless not avoid having non-empty inner layers and have to define ring subproblems anyway, so we decided to present an algorithm that is capable of solving ring subproblems. Another reason for this decision is that the additional complication is negligible.
\end{remark} 

We also want to be able to simplify ring subproblem \SP, this is to split \SP\ into ring subproblems $\SP_1, \ldots , \SP_{a}$ in such a way that the outer layer of each $\SP_i$ consists of only one polygon. This changes the problem not geometrically, but only composes it to smaller parts.
Luckily, we have to do that only for layer-unconstrained subproblems. It can be done also for layer-constrained problems, but is technically more demanding and requires the extra condition of small width.

Let \SP\ be a layer-unconstrained ring subproblem with outer layer $Q_\textup{out} = q_1,\ldots, q_a$, with $a\geq 2$. We define  \emph{layer-unconstrained} ring subproblems 
$\SP_1, \ldots , \SP_{a}$ as follows:
$\SP_i$ consists of $q_i$ and all connected components of $Q_\textup{in}(\SP)$ that are
contained in $q_i$ and likewise $P(\SP_{i})$ consists of all free points of $P(\SP)$ that are inside $q_i$.

\begin{figure}[tp]
 \centering
 \includegraphics{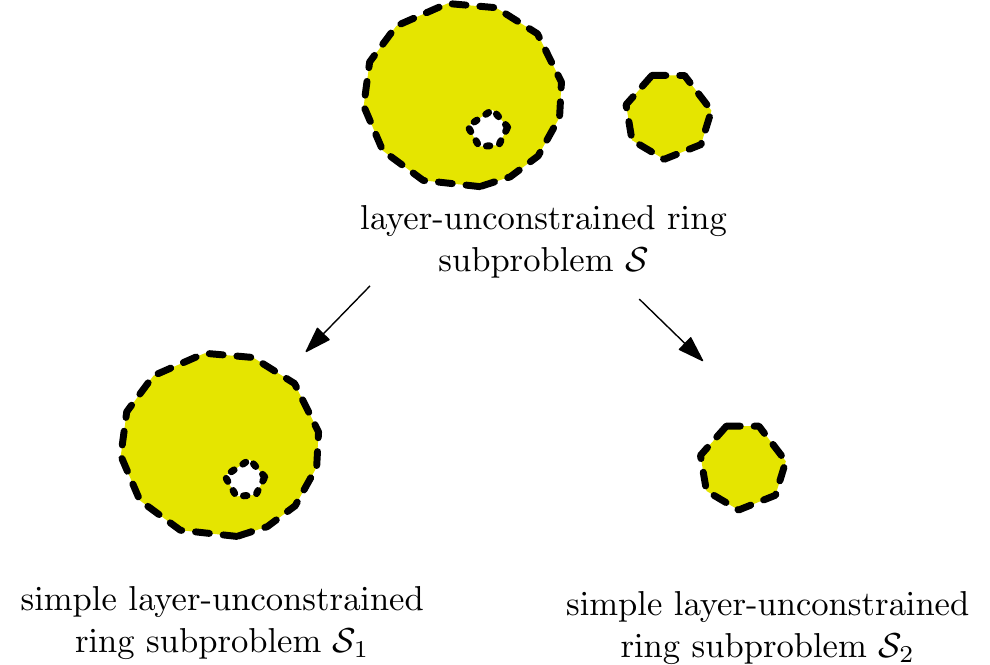}
 \caption{Illustration of the procedure \textsc{SplitByComponents}.}
\end{figure}

\begin{lemma}[Correct Counting of \textsc{SplitByComponents}]\label{lem:CorrectCountSimplify}
   Given a layer-unconstrained ring subproblem $\SP$ and assume it is split into $\SP_1, \ldots , \SP_{a}$. Then $\SP_1, \ldots , \SP_{a}$ are all simple layer-unconstrained ring subproblem and it holds:
   \[t(\SP) = t(\SP_1)\cdot \ldots \cdot t(\SP_a).\]
\end{lemma}
\begin{proof}
 Any valid triangulation $T$ of $\SP$ decomposes naturally into several valid triangulations
 $T_1,\ldots,T_a$ of $\SP_1, \ldots , \SP_{a}$ respectively. 
 Just restrict $T$ to $\SP_i$ to get $T_i$.
 
 Similarly the union of any collection of valid triangulations
 $T_1,\ldots,T_a$ of $\SP_1, \ldots , \SP_{a}$ is a valid triangulation of \SP.
 
 Note that the outerplanar index of $T_i$ 
 is at most the outer planar index of $T$, 
 but it could be smaller. 
 In the later case the inner 
 layer of $\SP_i$ is empty.
 This is the reason, it is convenient that 
 the width is only an upper bound on the number of layers.
 (See Condition~\ref{itm:CondInnerLayerRING} of Definition~\ref{def:ValidTriangRING}.)
\end{proof}


\begin{algorithm}[tp]
  \caption{\textsc{Peeling}: Counting layer-unconstrained ring problems}
  \label{alg:GeoRing}
  \begin{algorithmic}[1]
  \small
    \Require{Layer-unconstrained Ring Subproblem \SP}
    
    \State Initialize search tree $\tau$ to store all 
    computed values.
    \State \Return{\textsc{Count}(\SP)} 
    
    \Statex
    
    \Function{Count}{\SP}
      \If{\SP\ in $\tau$}
	\State \Return{$\tau(\SP)$}
      \ElsIf{\SP\ simple layer-unconstrained ring subproblem and $w(\SP)\geq m + 2$}
	  \State result = \textsc{SplitByLayers}(\SP)
      \ElsIf{\SP\ layer-unconstrained non-simple ring subproblem}
	  \State result = {\textsc{SplitByComponents}(\SP)}
      \ElsIf{\SP\ simple ring subproblem and $w(\SP)\leq  m +1$ }
	  \State result = {\textsc{Nibbling}(\SP)}
      \EndIf
      \State insert (\SP, result ) into $\tau$
      \State \Return{result}
    \EndFunction

    \Statex
    \Function{SplitByComponents}{\SP}
      \State Split \SP\ into $\SP_{1},\ldots,\SP_{a}$
      \State \Return{$\textsc{Count}(\SP_{1})\cdot \ldots \cdot \textsc{Count}(\SP_{a})$}
    \EndFunction

    \Statex
    \Function{SplitByLayers}{\SP}
    \State $s_1=0$
    \For{peripheral separation layer $L \in \LL(\SP)$}
      \State   $s_2 = 0$
      \State Define $\SP_\textup{out}(L)$ and $\SP_\textup{in}(L)$
	   by splitting $\SP$ using $L$ 
      \For{ constraint vector  $c \in \C_\textup{out}(L)$}
	  \State  $s_2 = s_2 \, + \, $\textsc{Count}($\SP_\textup{out}(L,c))$
      \EndFor
	\State  $s_1 = s_1 \, + \, $\textsc{Count}($\SP_\textup{in}(L)) \cdot s_2$
    \EndFor
      \State \Return{$s_1$}
    \EndFunction

  \end{algorithmic}
\end{algorithm}

We are now ready for the proof of Theorem~\ref{thm:GeoRingAlgo}.

\begin{proof}[Proof of Theorem~\ref{thm:GeoRingAlgo}]
  
  We will first describe the \textsc{Peeling} algorithm, then argue its correctness and finally supply a runtime analysis.
  
  The Algorithm \textsc{Peeling} uses the memoization technique and is based on dynamic programming. The subproblems of the dynamic programming scheme are the ring subproblems. Each ring subproblem is solved recursively, by either the \textsc{SplitByLayers}, \textsc{SplitByComponents} or \textsc{Nibbling} procedure. Which procedure will be called depends on the properties of the subproblem.  Already computed solutions are stored in and retrieved from the search tree $\tau$. This prevents repeated computations.  The pseudocode of 
  \textsc{Peeling} is depicted as Algorithm~\ref{alg:GeoRing}.
  
  To see correctness, we have to show that one of the cases in the main routine \textsc{Count} will be called. It is clear that each routine works correctly by Lemma~\ref{lem:CorrectCountSplitByLayers} and~\ref{lem:CorrectCountSimplify} and Theorem~\ref{thm:FullParaAlgo}.
  To be more precise \textsc{Count} will not have a valid value for the variable result if \textsc{Count} will be called with a non-simple constrained ring subproblem or with a constrained simple ring subproblem with width $w \geq m+2$.
  So we have to show that this will not happen.
  Initially, \textsc{Count} will be called by a layer-unconstrained ring subproblem.
  The subroutine \textsc{Nibbling} does not call \textsc{Count} at all. 
  The subroutine \textsc{SplitByComponents} does call \textsc{Count} only with layer-unconstrained simple ring subproblems. 
  The subroutine \textsc{SplitByLayer} defines a simple outer subproblem $\SP_\textup{out}$ and calls calls \textsc{Count} with it. But this is fine as the width of $w(\SP_\textup{out})\leq m$. The inner subproblem $\SP_\textup{in}$ is layer-unconstrained.
  This covers all cases and shows correctness.

  In order to bound the running time, we will bound the total number of possibly occurring subproblems.  Further, we bound the total costs for each subproblem. This is either the total number of recursive calls or the time spend by \textsc{Nibbling}.  
  As a first step we bound the number of annotated cactus graphs on $l$ points in the plane. Given a cactus graph on $l$ points $Q$ in the plane, it is well known that $G$ has at most $2l-3$ edges. Thus there are trivially at most \[\binom{l^2}{2l-3} + \binom{l^2}{2l-4} + \ldots +\binom{l^2}{1} + \binom{l^2}{0} \leq l^{4l- 6} + l^{4l- 8} + \ldots + 1 \leq l^{4l}\] cactus graphs on $Q$. Given a set $P$ of $n$ points there are \[\binom{n}{l}+\binom{n}{l-1} + \ldots + \binom{n}{0} \leq n^l + n^{l-1} + \ldots + 1 \leq n^{l+1}\] point sets $Q$ of size at most $l$. 
  Thus there are at most $n^l \cdot  l^{4l}$ cactus graphs of size at most $l$ on a set of $n$ points. 
  There are at most $|L|^{3l-3}$ ways to annotated all $l$ vertices and $2l-3$ edges.
  Given a set $S$ of points, every ring subproblem on $S$ is completely defined by the outer layer, the inner layer, and the inner layer  and outer layer index.
  If we choose $m= \lfloor \sqrt{n} \rfloor $ the total number of possible annotated inner layers equals $n^{(3 + o(1))\sqrt{n}} \cdot |L|^{(3 + o(1))\sqrt{n}}$ and the total number of outer layers is bounded by the same number. Note that the boundary of the initial layer-unconstrained ring subproblem might also appear as outer layer in some subsequent subproblems, but this is the only possible outer layer with potentially more than $\lceil \frac{n}{m}\rceil$ vertices.
  Thus the total number of subproblems is $n^{(6 + o(1))\sqrt{n}} \cdot |L|^{(6 + o(1))\sqrt{n}}$.
  The time to solve a subproblem called by \textsc{Nibbling} takes $n^{(5 + o(1))\sqrt{n}}\cdot |L|^{(6 + o(1))\sqrt{n}}$ time by Theorem~\ref{thm:FullParaAlgo}. The number of recursive calls by the procedure \textsc{SplitByComponents} is polynomial in $n$ and does not depend on $L$. The number of recursive calls by the procedure \textsc{SplitByLayers} equals $n^{(4 + o(1))\sqrt{n}}\cdot |L|^{(3 + o(1))\sqrt{n}}$. To see this recall that the set $\LL(\SP)$ has size $n^{(3 + o(1))\sqrt{n}}\cdot |L|^{(3 + o(1))\sqrt{n}}$ and the set $\mathcal{C}_\textup{out}(L)$ has size $n^{(1 + o(1))\sqrt{n}}$. Thus the total running time is bounded by $n^{(11 + o(1))\sqrt{n}}\cdot |L|^{(12 + o(1))\sqrt{n}} = n^{O(\sqrt{n})}\cdot |L|^{O(\sqrt{n})}$. 
\end{proof}

  A more careful analysis would probably give a slightly better constant. Possible improvements might be based on a different design of the algorithm, careful choice of $m$ and an improved estimate on the number of cactus graphs on $n$ vertices.


\section{Applications for Counting other Structures}
\label{sec:applications}

In this section, we give a complete description of the framework
discussed in Section~\ref{sec:applications} for counting
non-crossing straight line graphs. In Section~\ref{sec:tools}, we
explain formally how to go from counting triangulations to colored
graph classes.  As we show how to do this in full generality and not
just for a specific graph class, the description becomes a little
abstract and technical. We use the example of straight-line perfect
matchings to illustrate these definitions.  We attain the following result.
\begin{theorem}[Counting Perfect Matchings] \label{thm:CountPerfectMatchings}
  Given a set $S$ of $n$ points in the plane, there exists an algorithm that counts the total number of non-crossing perfect straight line matchings in $n^{O(\sqrt{n})}$ time.
\end{theorem}
Another key result of this
section is an operation $\oplus$ that can be used to combine
annotation systems in order to get the intersection of the graph
classes described by them.

In Section~\ref{sec:examples}, we apply this framework to a rich
selection of important graph classes. We describe annotation systems
for them and prove their correctness. This serves to show the power of
the framework and also explains how to use the framework. In
particular, we hope that certain tricks to design annotation systems
become clear and might be helpful to design annotation systems for
graph classes not covered by our examples.

We start with the proof of 
Theorem~\ref{thm:Count3Color} about counting $3$-colorable triangulations.
   
\begin{proof}[Proof of Theorem~\ref{thm:Count3Color}.]
  Given a triangulation $T$ that is $3$-colorable, this coloring is unique up to a permutation of the colors. It becomes unique as soon as the coloring of two adjacent vertices are fixed. Given a point set $S$ an edge $e^*=vw$ on $\partial CH(S)$, we define a list $L_\textup{3c}$ of feasible empty triangles such that for each triangle $\Delta$ holds
  \begin{enumerate}[noitemsep,topsep=3pt,parsep=3pt,partopsep=0pt]
   \item All three colors (red, yellow, blue) appear exactly once on the vertices of $\Delta$.
   \item $v$ is always red and $w$ is always yellow.
  \end{enumerate}
  We denote by $\mathcal{T}_3$ the set of $3$-colorable triangulations on $S$ and with $\mathcal{T}^A({L_{3c}})$ the set of annotated triangulations such that each triangle belongs to the list of $L_{3c}$.
  It holds 
  \[|\mathcal{T}^A({L_{3c}})| = |\mathcal{T}_3|.\]
  The size of this annotation system can be upper bounded by 
  $|L_{3c}|\leq \binom{n}{3}\cdot 6! = O(n^3) $.
  Now, we apply Theorem~\ref{thm:AnnotToAlgo} to count all triangulations with this annotation system. 
  \end{proof}

\subsection{Tools}\label{sec:tools}
In this section, we try to develop some general tools for annotation
systems.  Some of the definitions are a little abstract, so we try to
illustrate our definitions with the example of perfect matchings.  In
this section, we show how to count structures that are not
triangulations. We do this with the help of constrained Delaunay
triangulations. We are able to abstract out our arguments so that the
reader will not have to worry about constrained Delaunay triangulation
ever again after this section.  For this goal we develop a machinery
to combine annotation systems, which is useful in its own respect.

\begin{definition}
  Let $S$ be a set of $n$ points in the plane and let $c_V$ and $c_E$ be two integers.
  We define $\mathcal{G}_\textup{color}$ to be the set of non-crossing straight line graphs on $S$, such that 
   \begin{enumerate}[noitemsep,topsep=3pt,parsep=3pt,partopsep=0pt]
    \item each vertex has a color  $i\in \{1,\ldots,c_V \}$ and
    \item each edge has a color $i\in \{1,\ldots,c_E \}$.
  \end{enumerate}
  For later convenience, we assume that none of the colors is denoted by purple: the color purple is reserved for ``non-edges.''
  We call subsets $\mathcal{C} \subseteq \mathcal{G}_\textup{color}$ a \emph{classes of non-crossing colored graphs}. We also use the term \emph{property} to describe a class of colored graphs; we say $G$ has property $\mathcal{P}$, if $G\in \mathcal{P}$.
\end{definition}

An example of a colored graph class is the class of non-crossing straigth line perfect matchings $\mathcal{M}$ on $S$. This class has only one color, say orange. We want to keep this example in mind for the forthcoming somewhat technical definitions.

As the algorithm of Theorem~\ref{thm:AnnotToAlgo} counts annotated triangulations, we have to triangulate the matchings and we have to do it in a unique way. We will follow the ideas of previous authors and use the constrained Delaunay triangulation for this purpose.
See~\cite{DBLP:journals/dcg/AlvarezBCR15} and~\cite{DBLP:journals/comgeo/AlvarezBRS15} on how constrained Delaunay triangulations were used earlier in order to count perfect matchings and other crossing-free structures. The text book by Hjelle and D\ae hlen gives a detailed introduction to constrained Delaunay triangulations~\cite{hjelle2006triangulations} (Chapter 6).

\begin{figure}[t]
  \centering
  \includegraphics{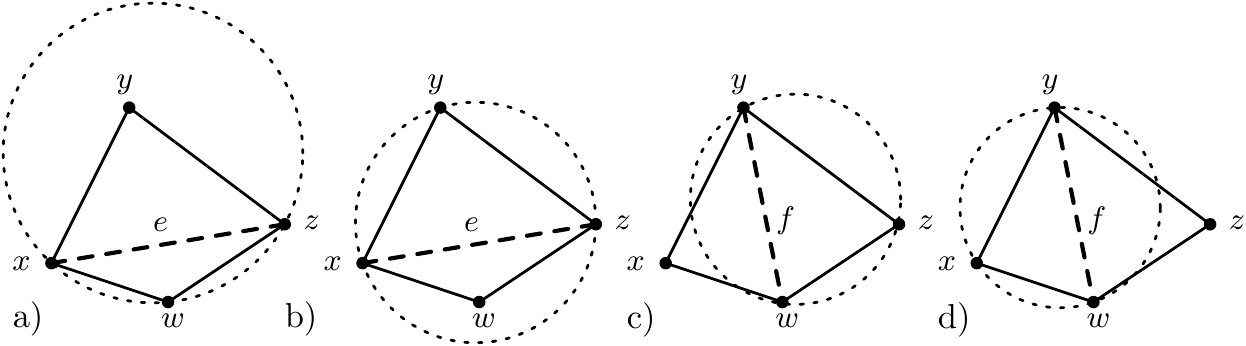}
  \caption{Diagonal $e$ of $xyzw$ does not satisfy the Delaunay condition: as shown in a) and b), circle $xzw$ contains $xyz$ and circle $xyz$ contains triangle $xzw$. On the other hand, diagonal $f$ does satisfy the Delaunay condition: as shown in c) and d), circle $yzw$ does not contain triangle $xyw$ and circle $xyw$ does not contain triangle $yzw$.}
  \label{fig:cdtDefinition}
\end{figure}

\begin{definition}[Constrained Delaunay Triangulation]\label{def:ConDelaunayTriang}
  Let $G\in \mathcal{G}_\textup{color}$ be a graph on a set $S$ of $n$
  points and a let $T$ be a triangulation extending $G$ on the same
  set of points.  We say an edge $e$ in a triangulation $T$ satisfies
  the \emph{Delaunay Condition} if the circumference of neither adjacent
  triangle contains the other adjacent triangle, see
  Figure~\ref{fig:cdtDefinition}. We say that $T$ is the
  \emph{constrained Delaunay triangulation} of $G$ if every edge
  $e \in E(T)\setminus E(G)$ satisfies the \emph{Delaunay Condition}.
  We also demand that no edge $e\in E(G)$ is colred purple and every
  other edge is colored purple.  (In case that $e$ is on the boundary
  of the convex hull of $S$, then we say that $e$ satisfies the
  Delaunay Condition as well.)  We denote by $CDT(\mathcal{P})$ the
  set of triangulations such that all purple edges satisfy the
  Delaunay Condition and the remaining colored graph belongs to
  $\mathcal{P}$.
\end{definition}
  See Figure~\ref{fig:TermsAnnotation} for an example of a constrained triangulation of a matching.
  \begin{figure}[p]
  \centering
  \includegraphics{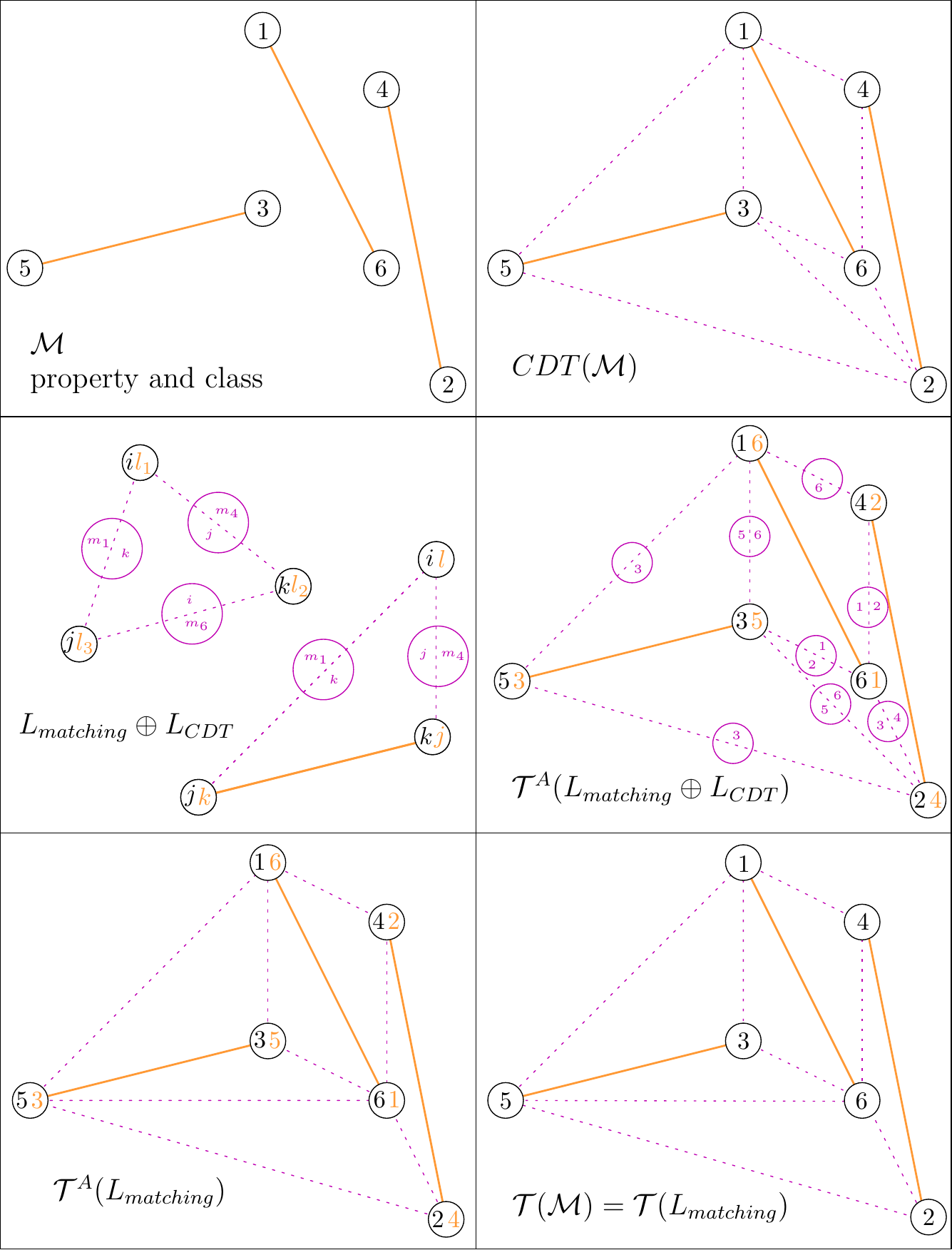}
  \caption{}
  \label{fig:TermsAnnotation}
\end{figure}
It is noteworthy that the Delaunay condition is symmetric. That is, let $\Delta_1$ and $\Delta_2$ be the two triangles adjacent to the edge $e$ and $C_1$ and $C_2$ be the respective circumferences. Then it holds that $C_1$ contains $\Delta_2$ if and only if $C_2$ contains $\Delta_1$. 
We assume that there are no $4$ points on a common circle. 
This can be achieved in polynomial time by a small affine transformation. 
The first algorithm to compute the constrained Delaunay triangulation in optimal $O(n\log n)$ time was given by Chew~\cite{DBLP:journals/algorithmica/Chew89}. We will use the fact that the constrained Delaunay triangulation is unique.

\begin{theorem}[\cite{hjelle2006triangulations}]
  Given $G \in \mathcal{G}$, there exists exactly one constrained Delaunay triangulation $T$ extending $G$.
\end{theorem}
We get the following corollary. 
\begin{corollary}[folklore]\label{cor:CountCDT} For every point set $S$ and 
every property $\mathcal{P}$ of colored graphs of $S$, we have $|\mathcal{P}| = |CDT(\mathcal{P})|$. 
\end{corollary}


Let us return to perfect matchings, that is, the property $\mathcal{M}$. The way we count $\mathcal{M}$ is to count $CDT(\mathcal{M})$. The most straightforward way would be to define an appropriate annotation system enforcing the matching property for the orange edges and the Delaunay Condition for the purple edges. However, we will go a different path and define two annotation systems. Later on, we will show how two annotation systems can be combined to a new one.

\begin{definition}[Annotations for Matchings]
  We want to define the annotation system $L_\textup{matching}$ in a way that all the orange edges are forced to form a matching. Consider a triangulation, where the orange edges form a perfect matching and imagine that every vertex $v$ is labeled by the unique vertex $w\neq v$ that is incident to the same orange edge. We define $L_\textup{matching}$ to be the set of all annotated triangles arising this way. 
  In particular, for every triangle $\Delta$ that annotates vertex $v$ with $w$, $\Delta$ must not block visibility between $v$ and $w$. 
  It is easy to see that the orange edges of any valid triangulation annotated by $L_\textup{matching}$ must form a perfect matching.
  
\end{definition}

\begin{definition}[Annotations for Delaunay-Condition]
  We want to define the annotation system $L_\textup{CDT}$ in a way that all the purple edges satisfy the Delaunay Condition. Consider a triangulation, where the purple edges satisfy the Delaunay Condition and imagine that every edge $e$ is labeled by the unique pair of vertices $v_1$ and  $v_2$ that are incident to the same triangles as $e$. (There might be only one such vertex in case that $e$ on $\partial CH(S)$.) We define $L_\textup{CDT}$ to be the set of all annotated triangles arriving in this way. It is easy to see that the purple edges of any valid triangulation annotated by $L_\textup{CDT}$ must satisfy the Delaunay Condition.
  It is easy to see that $|L_\textup{CDT}|\leq n^6$, as there are at most $\binom{n}{3}$ empty triangles and at most $n$  annotations are possible for each edge of each triangle (Note that there are two labels on an edge, but one of them is determined by the triangle itself).  
\end{definition}

We deliberatley do not specify the number of colors. 
We leave this open for later, as the description works for different number of colors.

There are several advantages of defining two annotation systems instead of just one. The first advantage is that we have to think only about one property at a time to show correctness. The second advantage is that we can reuse annotation systems, which is particularly interesting for $L_\textup{CDT}$.
The following definitions serve to define the combination of two annotation systems.

\begin{definition}[Color Annotation]
  Given integers $c_V$ and $c_E$, we define the \emph{complete color annotation} $L_\textup{color}$ of a point set $S$ as the set of all empty triangles $\Delta$ such that
  \begin{enumerate}[noitemsep,topsep=3pt,parsep=3pt,partopsep=0pt]
    \item each vertex is annotated by a number $i\in \{1,\ldots,c_V \}$ and
    \item each edge is annotated by a number $i\in \{1,\ldots,c_E \}$.
  \end{enumerate}
  We refer to these numbers usually as colors and us names like ``red,'' ``magenta,'' ``salmon'' and so on. In case $c_V$ or $c_E$ equals $0$, we say that the vertices or edges are are not colored.
  We also call these feasible triangles \emph{feasibly colored triangles}.
  The annotation system $L_\textup{color}$ has size $|L_\textup{color}| \leq \binom{n}{3}c_V^3 c_E^3$. 
\end{definition}

Note that we strictly distinguish between an edge colored by a specific color and an edge having the annotation of that color. Nevertheless, the intention of an annotation with a certrain color is to interpret this as that color.

The set $\mathcal{T}^A({L_\textup{color}})$ is exactly the set of triangulations of $S$ where each vertex and each edge is colored with a color from $\{1,\ldots, c_V\}$ and $\{1,\ldots, c_E\}$ respectively.

\begin{definition}[Extensions]
  We say that an annotation system $L$ is an \emph{extension of a color annotation} if each annotation consists of two components and $L$ restricted to the first component is a subset of the color annotation. The first component is also called the \emph{color component}. In the following, we restrict ourselves of annotation systems that are extensions of the color annotation $L_\textup{color}$.
\end{definition}

We assume that the numbers $c_V$ and $c_E$ are implicitly clear. Usually we omit to specify them as the description for the annotation system would work for various values of $c_V$ and $c_E$.
Further, we will assume that $c_V$ and $c_E$ are constants and thus they only add a multiplicative constant to the total size of any extension of the color annotation.

Going back to our matching example, annotation system $L_\textup{matching}$ can be regarded as an extension of $L_\textup{color}$, where we have two edge colors and no vertex color. We denote the two colors with orange and purple.

\begin{definition}[Annotations Describing Properties]
  Given an annotation system $L$, it defines a 
  set of valid annotated triangulations
  ${\mathcal{T}^A}(L)$. 
  When we restrict the triangulations 
  $T\in {\mathcal{T}^A}(L)$ to the colored 
  component, we receive a \emph{multi-set} of 
  colored triangulations denoted by $\mathcal{T}_L$. 
  To do this, we interpret the color annotation as the color itself.
  If each colored triangulation  appears only once, 
  it holds $\mathcal{T}(L) \subseteq \mathcal{G}_\textup{color}$.
  Given a property $\mathcal{P}$, we define $\mathcal{T}(\mathcal{P})$ to be the set of  all those triangulations $T$ such that 
  the graph $G_T$ that remains after removing all purple edges satisfies $\mathcal{P}$. 
  We say that an annotation system $L$ \emph{describes} a property $\mathcal{P}$, if $\mathcal{T}(\mathcal{P}) = \mathcal{T}(L)$ holds.
 (Note that we assume here that $L$ has one edge color more than $\mathcal{P}$.
 This additional color is denoted by purple.)
\end{definition}
 
 
\begin{definition}[Combining Annotations]
  Let  $L_1$ and $L_2$ be annotation systems extending $L_\textup{color}$. We define \[L = L_1 \oplus L_2\] as the set of triangles such that each annotation of a vertex and edge has three components $(c_1,c_2,c_3)$ satisfying the following properties:
  \begin{enumerate}[noitemsep,topsep=3pt,parsep=3pt,partopsep=0pt]
    \item each triangle restricted to the first component $c_1$ is feasibly colored,
    \item each triangle restricted to the first and second component $(c_1,c_2)$ belongs to $L_1$, and
    \item each triangle restricted to the first and third component $(c_1,c_3)$ belongs to $L_2$.
  \end{enumerate}
In the following, whenever we will use the $\oplus$ operation, we will \emph{implicitly} assume that both annotation systems extend the same colored annotation $L_\textup{color}$. Then $L$ can be regarded again as an extension of $L_\textup{color}$ by regarding the second and third component as one ``meta-component.''
\end{definition}

\begin{theorem}\label{lem:Intersection}
  Let $\mathcal{P}_1$ and $\mathcal{P}_2$ be two properties described by $L_1$ and $L_2$ respectively. Then $\mathcal{P}_1 \cap \mathcal{P}_2 $ is described by $L_1\oplus L_2$, and it holds $|L_1 \oplus L_2| \leq |L_1| \cdot |L_2|$.
\end{theorem}
\begin{proof}
 We will show
  \[ \mathcal{T}({\mathcal{P}_1 \cap \mathcal{P}_2}) 
  \stackrel{ (1)}{=} 
  \mathcal{T}({\mathcal{P}_1}) \cap \mathcal{T}({\mathcal{P}_2}) 
  \stackrel{ (2)}{=} 
  \mathcal{T}({L_1})\cap \mathcal{T}({L_2}) 
  \stackrel{ (3)}{=}
  \mathcal{T}({L_1 \oplus L_2}).\]
  
  We start with $(1)$. Let $T\in \mathcal{T}({\mathcal{P}_1 \cap \mathcal{P}_2})$ and $G$ be the graph remaining after removing all purple edges. Then $G$ satisfies $\mathcal{P}_1$ and $\mathcal{P}_2$. Thus $T \in \mathcal{T}({\mathcal{P}_1})$  and $T \in \mathcal{T}({\mathcal{P}_2})$. 
  This shows ``$\subseteq$''. Conversely, let $T\in \mathcal{T}({\mathcal{P}_1}) \cap \mathcal{T}({\mathcal{P}_2})$ 
  and $G$ be the graph that is left after removing all purple edges. Then $G$ satisfies $\mathcal{P}_1$ and $\mathcal{P}_2$ and thus $T$ belongs to $\mathcal{T}({\mathcal{P}_1 \cap \mathcal{P}_2})$. This shows $(1)$. 
  
Statement $(2)$ follows immediately from the assumption $\mathcal{T}({\mathcal{P}_i}) = \mathcal{T}({L_i})$, for $i = 1,2$.
  
  We finish with $(3)$. 
  Let $T\in\mathcal{T}({L_1})\cap \mathcal{T}({L_2})$ be a colored triangulation. 
  Then there exists two annotations $A_1$ and $A_2$ such that $T$ annotated with $A_i$ belongs to ${\mathcal{T}^A}({L_i})$, for $i=1,2$. Now, we annotate $T$ with $A_1$ \emph{and} $A_2$ carefully. So let $T'$ be an annotated triangulation such that the first component represents its color, the second component is given by $A_1$ and the third component is given by $A_2$. By definition $T'$ belongs to ${\mathcal{T}^A}({L_1\oplus L_2})$. Thus $T$ belongs to $\mathcal{T}({L_1\oplus L_2})$. This shows ``$\subseteq$''.
  The reverse direction is easier. Let $T \in \mathcal{T}({L_1\oplus L_2})$ be a colored triangulations corresponding to the annotated triangulation 
  $T' \in {\mathcal{T}^A}({L_1\oplus L_2})$. 
  Then by definition of $\oplus$, $T'$ restricted to the second component belongs to ${\mathcal{T}^A}({L_1})$ and 
  $T'$ restricted to the third component belongs to ${\mathcal{T}^A}({L_2})$.
  Thus $T'$ belongs to ${\mathcal{T}^A}({L_1}) \cap {\mathcal{T}^A}({L_2})$.
  
  The size bound follows easily from the construction.
\end{proof}

%

\begin{lemma}
  Let $\mathcal{P}$ be a property described by an annotation system $L$. 
  Then \[|\mathcal{P}| = |{\mathcal{T}^A}({L\oplus L_\textup{CDT}})|.\]
\end{lemma}
\begin{proof}
  From Corollary~\ref{cor:CountCDT} follows $|\mathcal{P}| = |CDT(\mathcal{P})|$.
  We know that $\mathcal{T}({L})$ contains precisely the graphs in $\mathcal{P}$
  and thus ${\mathcal{T}^A}({L})$ represents precisely all possible triangulations of all graphs of $\mathcal{P}$ with purple edges. 
  Further, we know that all purple edges of ${\mathcal{T}^A}({L\oplus L_\textup{CDT}})$
  satisfy the Delaunay Condition, which implies, that each graph in $\mathcal{P}$ is uniquely triangulated.
\end{proof}

\begin{theorem}[Counting Colored Graph Classes]\label{thm:CountColoredGraphClasses}
    Let $\mathcal{P}$ be a colored graph class and $L$ an annotation system describing it. There exists an algorithm that counts the cardinality of $\mathcal{P}$ in $n^{O(\sqrt{n})} \cdot |L|^{O(\sqrt{n})}$.
\end{theorem}
\begin{proof}
    In order to count the members of $\mathcal{P}$, it is sufficient to count ${\mathcal{T}^A}({L\oplus L_\textup{CDT}})$. It remains to prove the upper bound on the running time. From Theorem~\ref{thm:AnnotToAlgo}, we have the bound of $n^{O(\sqrt{n})} \cdot |L\oplus L_\textup{CDT}|^{O(\sqrt{n})}$ on the running time. Note that $|L\oplus L_\textup{CDT}| \leq |L|\cdot |L_\textup{CDT}| \leq |L|\cdot n^6$. Thus the whole running time is bounded by $n^{O(\sqrt{n})} \cdot |L|^{O(\sqrt{n})} \cdot |L_\textup{CDT}|^{O(\sqrt{n})} \leq n^{O(\sqrt{n})} \cdot |L|^{O(\sqrt{n})}$.
\end{proof}

  To see the usefulness of Theorem~\ref{thm:CountColoredGraphClasses}, observe that now the proof of Theorem~\ref{thm:CountPerfectMatchings} boils down to the fact that $L_\textup{matching}$ describes the class of orange perfect matchings and has polynomial size.

\subsection{Examples}\label{sec:examples}

  \begin{definition}[Degree Constrains]
    Given a family $\mathcal{D} = \{D_v : v\in V \}$ of sets $D_v \subseteq \{1,\ldots, n \}$. We say a triangulation \emph{satisfies the degree constraint} given by $\mathcal{D}$ if  $\deg(v)\in D_v$ for each $v\in V$. And a triangulation $T$ satisfying the degree constrained imposed by $\mathcal{D}$ is called \emph{degree constrained triangulation}.
    
    We say a colored graph \emph{satisfies the colored degree constraint in color $i$} given by $\mathcal{D}$ if  $\deg_i(v)\in D_v\in\mathcal{D}$ for each $v\in S$ (The expression $\deg_i(v)$ denotes the number of edges in color $i$ incident to $v$.). And a graph $G$ satisfying the colored degree constraints imposed by $\mathcal{D}$ is called \emph{colored-$i$ degree constrained graph}.  
  \end{definition}

%

\begin{lemma}\label{lem:Annotations}
  There are polynomial sized annotation systems for the following colored non-crossing graph classes:
  \begin{enumerate}[noitemsep,topsep=3pt,parsep=3pt,partopsep=0pt]
   \item $L_\textup{triang}$ colored triangulations
   \item $L_\textup{TriDeg}$ degree constrained triangulations
   \item $L_\textup{Deg}$ colored degree constrained non-crossing graphs
   \item $L_\textup{SpanTree}$ red non-crossing spanning tress
   \item $L_\textup{Connect}$ blue connected non-crossing graphs
   \item $L_\textup{Size}$ $k$-green vertices
   \item $L_\textup{Ind}$ orange vertices form an independent set
   \item $L_\textup{Dom}$ brown vertices form a dominating set
   \item $L_\textup{FaceDeg}$ face-degree constrained non-crossing graphs
  \end{enumerate}
\end{lemma}
\begin{proof}
   The way, we describe the annotation systems always follows the same scheme. We start with a graph class $\mathcal{P}$ that we want to enforce. Then, we describe a procedure to annotate it in an \emph{unambiguous} way. The annotation system is implicitly defined by all possible annotated empty triangles, that might arise in this way. We will \emph{usually} not make it explicit, as it is only technical and not enlightening. See Figure~\ref{fig:TermsAnnotation} for an explicit example for perfect matchings. We have to show that the annotation system describes $\mathcal{P}$. For this step, it is necessary and sufficient that the annotation procedure is unambiguous and enforces the property $\mathcal{P}$. The last step is to give a size bound on the size of the annotation system.
   
   Recall that the order label was attached to the points as a preprocessing step. This labeling helps to set priorities between vertices and makes possible choices unambiguous.
   
  \begin{description}
  
  \item[Colored Triangulations] Define $L_\textup{triang}\subseteq L_\textup{color}$, as the list of all those annotated triangles, that do not use the color purple reserved for ``non-edges''.
  
   \begin{figure}[t]
    \centering
    \includegraphics{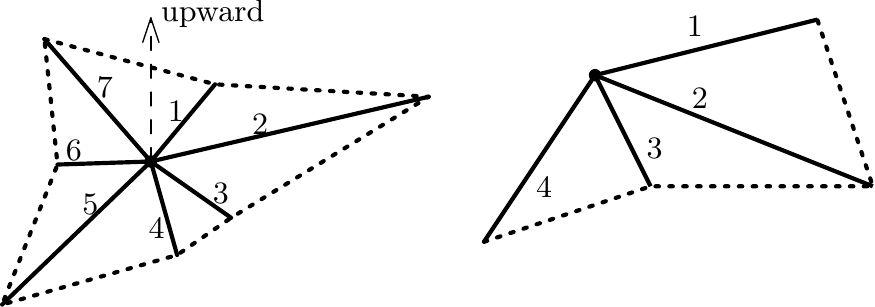}
    \caption{}
    \label{fig:AnnotationDegree}
  \end{figure}
  
   \item[Degree Constraints] We start to show how to enforce the degree on one vertex.
   Let $T$ be a triangulation and $v$ a vertex of $T$ with degree $d$. We annotate the edges around $v$ in clockwise order with $1,\ldots , d$.
   We start with the edge, which comes first after the upward direction.
   This is unambiguous and enforces degree $d$. See also Figure~\ref{fig:AnnotationDegree}. Note that the triangle containing the upward direction has the annotation $d$ and $1$ in clockwise order around $v$. 
   In case that $v$ is on $\partial CH(S)$, it is even simpler, as we start and end our counterclockwise numbering on $\partial CH(S)$. This is we force the edges incident to $v$ to carry the annotations $1$ and $d$ respectively.
   
   Given a set $D\subset \{1,\ldots, n\}$, we can take the union of all possible annotations enforcing a degree $d\in D$. The size of this annotation system is still pretty small, as every triangle adjacent to $v$ gets added only two numbers, where the second number is defined by the first one. 
   
   Now, given a set $\mathcal{D}$ as in the definition of degree constrained triangulations, we can define an annotation for each vertex as above. We can combine all of these annotations without increasing the size too much, as every triangle has to remember only two numbers for each of its three vertices. This amounts to a total of $O(n^{3+6})$ many triangles.

   \begin{figure}[t]
    \centering
    \includegraphics{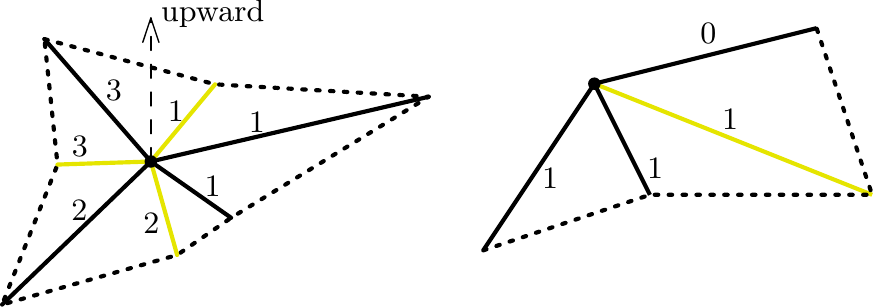}
    \caption{}
    \label{fig:AnnotationColoredDegree}
  \end{figure}
   \item[Colored Degree Constraints] 
   We show how to enforce a yellow degree constrained on one vertex. 
   The generalization to more possible degrees and all vertices simultaneously work in the same way as for uncolored degree constrained triangulations.
   Generalizing it to more colors can be done using the $\oplus$ operation.
   Given a colored graph $G$ we can assume that $G$ is completely triangulated by edges not belonging to $G$ (purple edges). 
   Let $v \in V(G)$ be a vertex with  $d$ yellow incident edges.
   We annotate each edge(also non-yellow edges) with the number of yellow edges in the counterclockwise interval between the upwards direction and this edge, including the edge itself, see Figure~\ref{fig:AnnotationColoredDegree}.
   In case that $v$ is on $\partial CH(S)$, it is even simpler, as we start and end our counterclockwise numbering on $\partial CH(S)$.
   
   If more than one color should be constraint the corresponding annotation systems can be combined with the $\oplus$ operator. Using the colors, the resulting size of the annotation system is $n^O(c)$.

   \begin{figure}[t]
    \centering
    \includegraphics{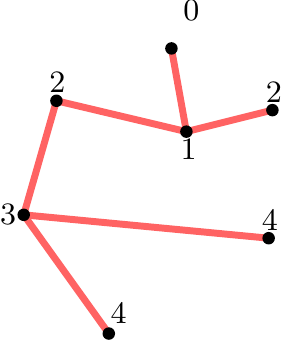}
    \caption{}
    \label{fig:AnnotationSpanningTree}
  \end{figure}
   
   \item[Spanning Trees] Let $G$ be a red spanning tree of $S$. Then we can think of $G$ as a rooted tree. We choose the root to be the vertex with lowest order label. We annotate each vertex $v$ with the 
   distance to the root and the label of its parent. It is easy to see that this gives an annotation in an unambiguous way.
   The set of feasible triangles defined by this annotation procedure does not include triangles $\Delta$, where a vertex $v$ is annotated with a vertex $w$, but the vision between $v$ and $w$ is blocked by $\Delta$.
   We will show, it also enforces the red edges to form a spanning tree. Every vertex, except the root, has exactly one vertex with lower distance to the root. This is enforced by the label on the vertex. Further following these edges defines a path to the root. This shows connectedness. At last, we have to show that this annotation system excludes cycles. To see this, assume, for the purpose of contradiction there is a cycle $C$ in $G$ and let $v\in V(C)$ be the vertex with largest  distance to the root. By the discussion above it has at most one neighbor with smaller distance to the root. a contradiction to the fact that every vertyex in $C$ has two neighbors.
   See Figure~\ref{fig:AnnotationSpanningTree}.
   
   To bound the size of $L_\textup{SpanTree}$ note that every vertex is annotated with two numbers and an upper bound of $O(n^{3+6})$ follows.

   \begin{figure}[t]
    \centering
    \includegraphics{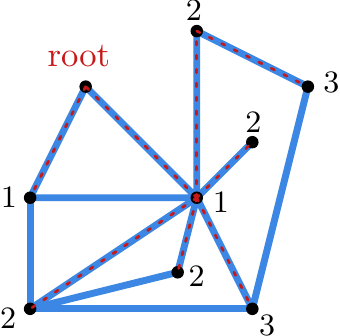}
    \caption{}
    \label{fig:AnnotationConnectedness}
  \end{figure}
   
   \item[Connectedness] Let $G$ be a blue connected graph. We will define a spanning tree $T_G$ in an unambiguous way. For this recall that we can assume that all vertices have an order label. This order label is just a number from $1$ to $n$ and gives arbitrary preferences. We will use it to single out a particular spanning tree among all spanning trees of $G$. It is important to keep in mind that this order label is \emph{not} part of the annotation, but given beforehand.
   
   The root of our spanning tree is the vertex with lowest order label. We annotate each vertex with the distance to the root. And we annotate each vertex with the neighbor that has the smallest order label among all neighbors in $G$ closer to the root.

   Now every vertex has a unique neighbor closer to the root and this neighborhood relation defines unambiguously a spanning tree of $G$.
   See Figure~\ref{fig:AnnotationConnectedness}.
   
   To avoid confusion let us be more explicit in this case, which
   properties annotated empty triangles of $L_\textup{Connect}$ must have.
   Well, every vertex is annotated with its distance to the root and its parent. Further, there better be a blue edge between each vertex $v$ and its parent $p(v)$. Also the parent should have smaller order label.
   It is maybe more subtly to see, that there must not be a blue edge between a vertex $v$ and another vertex $w$ that has smaller order label than $p(v)$.
   Further, a triangle $\Delta$ is not feasible, if there exists a vertex $v\in \Delta$, which is annotated with a vertex $w$, but the vision between $v$ and $w$ is blocked by $\Delta$.
   
   The size of $L_\textup{Connect}$ is bounded by $O(n^{3+6})$.
   
  \begin{figure}[t]
    \centering
    \includegraphics{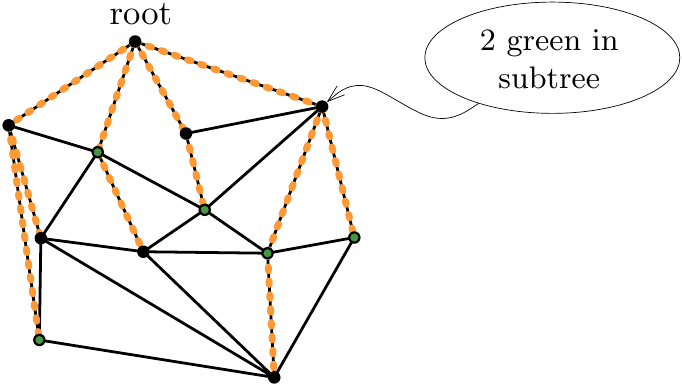}
    \caption{}
    \label{fig:AnnotationCount}
  \end{figure}
   
   \item[Count Vertices] Given a triangulation $T$ on $S$, we want to ensure that $T$ has exactly $k$ vertices colored green. We ignore all other colors for this description. We have seen in {\bf Connectedness} how we can identify a spanning tree $U_T$ in $T$ in an unambiguous way and find an annotation identifying $U_T$. Additionally, we annotate each vertex with the number of green vertices in its subtree. 
   We can interpret this number as the weighted indegree within $U_T$.
   In order to keep the information consistent we use the annotation system of colored degree, to check that the actual indegree adds up to the number written on each vertex.
   See Figure~\ref{fig:AnnotationCount}.
   
   Correctness follows from the correctness of the previous constructions. 
   The polynomial size bound follow in the same way.
   
   \item[Independent Set]
    Given a graph $G$ colored green with some vertices being colored orange, forming an independent set. 
    This means that there is no edge incident to two brown vertices. 
   
    Thus $L_\textup{Ind}$ consists of all feasible triangles except those which have two adjacent brown vertices.
   
   \item[Dominating Set]
    Given a graph $G$ colored green with some vertices being colored brown, forming a dominating set. We annotate every non-brown vertex with the adjacent brown vertex with smallest order label. 
   
   $L_\textup{Dom}$ consists of all feasible triangles arising in this way.
   In particular, consider the case that a vertex $v$ is annotated with $w$.
   then $v$ must not be part of a triangle $\Delta$ blocking visibility between $v$ and $w$. 
   Further, $v$ should not be adjacent to a brown vertex with lower order label and any feasible triangle with $v$ and $w$ should color $w$ brown in case that $v$ is annotated with $w$.
   
   The size of $L_\textup{Dom}$ is bounded by $O(n^{3+3})$, as every vertex is annotated with only one number.

   \begin{figure}[t]
    \centering
    \includegraphics{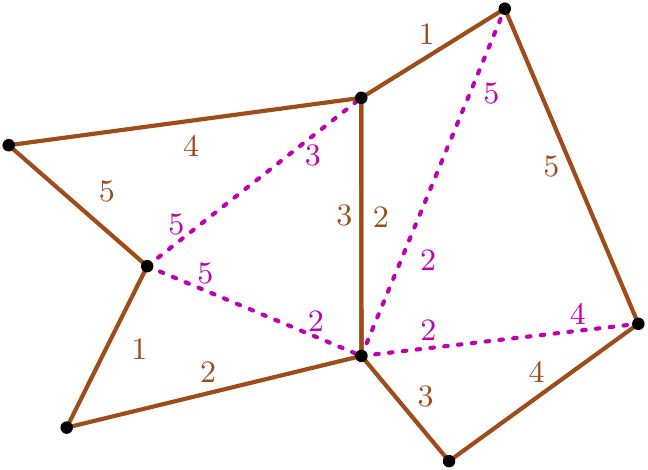}
    \caption{}
    \label{fig:AnnotationFaceDegree}
  \end{figure}
   
   \item[Face Degree Constraints] Given a brown graph $G$ with every bounded face of degree $d$. Here, we need again to make use of the triangulation $T$ containing $G$.
   Consider one specific face $f$ surrounded by the vertices $v_1,\ldots, v_d$ and edges $e_1,\ldots,e_d$ with $e_i = \{v_i,\ldots, v_{i+1}\}$ in counterclockwise order. Then there exists a vertex with lowest order label; without loss of generality assume it is $v_1$. Annotate every brown edge with $1$.
   We annotate edge $e_i$ with $i$, meaning that it is the $i$-th edge of this face. Further, we annotate the edges inside $f$ around $v_i$ with $i$. (These edges do not belong to $G$.)
   In this way, every brown edge receives three annotations(being incident to two faces) and every other edge receives also two annotations, as it is incident to two vertices. 
   
   The annotation system $L_\textup{FaceDeg}$ consists of all triangles arising in this way. 
   
   The size of $L_\textup{FaceDeg}$ is bounded by $O(n^{3+9})$, as every edge is annotated with at most three numbers. See Figure~\ref{fig:AnnotationFaceDegree}.
  \end{description}
  This finishes the proof. 
\end{proof}

We are now ready to proof Theorem~\ref{thm:CountingGeometricStructures}.
For convenience, we restate the theorem. 

\CountingGeometricStuff*

\begin{proof}
    The proof is just a combination of annotation systems that are given in Lemma~\ref{lem:Annotations}.
\begin{description}[noitemsep,topsep=3pt,parsep=3pt,partopsep=0pt]
   \item[All Graphs]
   To count the set of all graphs, we use an annotation system with exactly two edge colors and two vertex colors, namely red and purple. We enforce that all purple edges satisfy the constrained Delaunay Condition. 
   \item[Perfect Matchings]
    Enforce degree one on all vertices.
   \item[Cycle Decompositions]
    Enforce that all vertices have degree two.
   \item[Hamilton cycles, Hamilton paths]
    For the first ensure that all vertices have degree two and the graph is connected. For the second ensure that the graph has max degree two and is connected.
   \item[Euler-Tours]
    Ensure that the graph has even degree.
   \item[Spanning Trees]
    See Lemma~\ref{lem:Annotations}.
   \item[$d$-regular graphs]
    See Lemma~\ref{lem:Annotations}.
   \item[Quadrangulations]
    Enforce that each face has degree four and that the graph is spanning. \qedhere
  \end{description}
\end{proof}

\begin{remark}[On the scope of the framework]
 
The scope of the structures that we are capable 
to count is not limited by the examples we have given. 
It would be nice to encapture more precisely to which 
structures our procedure applies and to which it doesn't. 
However, there are some very simple structures that we 
were not able to give a descriptive annotation system for. 
Examples are $3$-colorable graphs, graphs with diameter 
exactly $k$ and graphs with degree constraints on more 
than constantly many colors. 
The problem of $3$-colorable graphs is that there might be 
more than one $3$-coloring and we don't know how to ensure 
that we count each graph only once. 
The diameter seems to be a too complex graph property. 
Known algorithms to compute the diameter usually compute the 
distance of all pairs of vertices. If someone finds a way how 
to compute the diameter of a planar graph in a simpler way, 
this algorithm might also lead to an appropriate annotation system.
The reason, we don't know how to give degree constraints to more than a constant number of colors with the above running time is that, we have to maintain too many informations on the edges to keep the informations consistent. 

\end{remark}

\begin{remark}[On the geometry]
  In general all structures that we are interested in are purely defined by combinatorial means and not so much by geometry. To be more precise: given a set of points $S$ the \emph{order type} of $S$ is defined as an oracle that returns for every triple of points whether they are oriented in clockwise or counterclockwise order. This information is considered as combinatorial, as it does not depend on the concrete coordinates of the set of points. In particular there is only  a finite number of possible order types for $n$ points in the plane.
  Nevertheless, we make use of the geometry at several occasions, in order to make possible choices unique. It is not clear to us, if this could be avoided. 
\end{remark}


\paragraph*{Acknowledgment}
Both authors are supported by the ERC grant PARAMTIGHT: 
Parameterized complexity and
the search
for tight complexity results", no. 280152.

\bibliographystyle{abbrv} 
\bibliography{LibSoCG}

\end{document}